\documentclass[11pt]{article}
\usepackage{graphicx}
\usepackage[misc,geometry]{ifsym}
\usepackage{verbatim}
\usepackage{amsfonts}
\usepackage{caption}
\usepackage{tabularx}
\usepackage{array}% http://ctan.org/pkg/array
\usepackage{booktabs}% http://ctan.org/pkg/booktabs
\usepackage{multirow}
\usepackage{fancyvrb}
\usepackage{amsmath}
\usepackage{color,colortbl}

\usepackage{enumerate}

\usepackage{latexsym}
\usepackage{amsmath,amssymb}

\def\a{\alpha} \def\b{\beta} \def\d{\delta} \def\D{\Delta}
\def\e{\epsilon}    \def\g{\gamma}
\def\ep{\epsilon}

   \def\Th{\Theta}  \def\l{\lambda}
 \def\m{\mu}  
\def\r{\rho}
  
\def\t{\tau} \def\om{\omega}  \def\Om{\Omega}

\def\cY{{\cal Y}}
\def\cC{{\cal C}}
\def\cG{{\cal G}}

\def\cX{{\cal X}}

\newcommand{\ul}[1]{\mbox{\boldmath$#1$}}

\newcommand{\brac}[1]{\left(#1\right)}

\newcommand{\bfrac}[2]{\left(\frac{#1}{#2}\right)}

\newcommand{\rai}{\rightarrow \infty}

\newcommand{\set}[1]{\left\{#1\right\}}

\def\seq{\subseteq}

\def\E{\mathbf{E}}

\def\Pr{\mbox{{\bf Pr}}}
\def\whp{{w.h.p.}}

\newcommand{\beq}[1]{\begin{equation}\label{#1}}
\newcommand{\eeq}{\end{equation}}

\newenvironment{proof}{\trivlist\item[]\emph{Proof}.}%
{\unskip\nobreak\hskip 1em plus 1fil\nobreak$\Box$
\parfillskip=0pt%
\endtrivlist}

{\unskip\nobreak\hskip 1em plus 1fil\nobreak$\Box$
\parfillskip=0pt%
\endtrivlist}
{\unskip\nobreak\hskip 2em plus 1fil\nobreak$\Box$
\parfillskip=0pt%
\endtrivlist}
{\unskip\nobreak\hskip 2em plus 1fil\nobreak$\Box$
\parfillskip=0pt%
\endtrivlist}

\newcommand{\ignore}[1]{}

\newtheorem{theorem}{Theorem}
\newtheorem{lemma}{Lemma}
\newtheorem{corollary}{Corollary}

\begin{document}

\title{The power of two choices in distributed voting%
\thanks{
This work was partially supported by EPSRC grant EP/J006300/1, ``Random Walks on Computer Networks'', 
the Austrian Science Fund (FWF) under contract P25214-N23 
``Analysis of Epidemic Processes and Algorithms in Large Networks'', and
the 2012 SAMSUNG Global Research Outreach (GRO) grant ``Fast Low Cost Meth-
ods to Learn Structure of Large Networks.''
}}

%\titlerunning{The power of two choices in distributed voting}        % if too long for running head

\author{
Colin Cooper\thanks{Department of Informatics, King's College London, UK.
{\tt colin.cooper@kcl.ac.uk}}
\and Robert Els\"asser\thanks{Department of Computer Sciences, University of Salzburg, Austria.
{\tt elsa@cosy.sbg.ac.at}}
\and Tomasz Radzik\thanks{Department of Informatics, King's College London, UK.
{\tt tomasz.radzik@kcl.ac.uk}}
}

%\author{Colin Cooper\inst{1}
%\and Robert Els\"asser\inst{2}
%$\and Tomasz Radzik\inst{1}}

%\authorrunning{Short form of author list} % if too long for running head

%\institute{Department of Informatics,
%              King's College London, 
%              United Kingdom \\
%             \email{colin.cooper@kcl.ac.uk,  tomasz.radzik@kcl.ac.uk} 
%\and
%Department of Computer Sciences, University of Salzburg, Austria \\
%\email{elsa@cosy.sbg.ac.at}}

\maketitle

\begin{abstract}
Distributed voting is a fundamental topic in distributed computing.
In the standard model of pull voting,
in each step every vertex chooses a
neighbour uniformly at random,
and adopts the opinion of that neighbour.
The voting is  said to be completed
when all vertices  hold the same opinion.
On many graph
classes including regular graphs,
irrespective of the expansion properties,
 pull voting requires
 $\Omega(n)$ expected time steps to complete,
even if initially there are only two distinct opinions
with the minority opinion being sufficiently large.

In this paper we consider a related  process which we call two-sample voting.
In this process every vertex chooses
two random neighbors in each step. If the opinions of these neighbors coincide,
then the vertex revises its opinion according to the chosen sample. Otherwise,
it keeps its own opinion.
We consider the performance of this process
in the case where  two different opinions reside on vertices
of some (arbitrary) sets $A$ and $B$, respectively. Here, $|A|+|B|=n$ is
the number of vertices of the graph.

%Assume that the initial imbalance between the sets is $|A| -|B| = \nu_0 n$.
We show that there is a constant $K$ such that if
the initial imbalance between the two opinions is
$\nu_0 = (|A| -|B|)/n \ge K\sqrt{(1/d) + (d/n)}$, then with high 
probability two sample voting 
completes in a random $d$ regular graph
in $O(\log n)$ steps and the initial majority opinion wins.
We also show the same performance 
for any regular graph,
if $\nu_0 \ge K\lambda_2$, where $\lambda_2$
is the second largest
eigenvalue of the transition matrix.
In the graphs we consider, standard pull voting requires $\Omega(n)$ steps, and the
minority can still win with probability $|B|/n$.
Our results  hold even if an adversary is able to rearrange the opinions in each
step, and has complete knowledge of the graph structure.

\end{abstract}

%\thispagestyle{empty}

%\newpage

%\setcounter{page}{1}

\section{Introduction}

Distributed voting has  applications in various
fields including consensus and leader election in large networks
\cite{BMPS04,HassinPeleg-InfComp2001}, serialisation of read/write %operations
in replicated data-bases \cite{Gifford79}, and the analysis of social
behaviour in game theory \cite{DP94}. Voting algorithms are usually simple,
fault-tolerant, and easy to implement \cite{HassinPeleg-InfComp2001,Joh89}.

One straightforward form of distributed voting is {\em pull voting}.
In the beginning each vertex of a connected undirected graph $G=(V,E)$ 
has an initial opinion.
The voting process proceeds synchronously in discrete time steps  called rounds.
During each round, each vertex independently contacts a random neighbour
and adopts the opinion of that neighbour.
The completion time $T$ is the number of
rounds needed for a single opinion
to emerge. 
This time depends on the structure of the underlying graph, and
is normally measured in terms of its expectation $\E T$.
We showed in \cite{CEHR-SIAM2013} that 
with high probability
%for the case that the initial opinions are distinct, 
the completion time is
$O(n / (\nu (1-\lambda_2))$, where $n$ is the number of vertices, $\lambda_2$ is the
second largest eigenvalue of the transition matrix, $\nu = \sum_{v \in V}
d^2(v)/(d^2 n)$, $d(v)$ is the degree of vertex $v$ and $d$ is the average degree.

In the {\em two-party voter model}, vertices initially hold one of two opinions $A$ and $B$.
As usual, the pull voting is completed when all vertices have the same opinion.
Hassin and Peleg~\cite{HassinPeleg-InfComp2001}
and Nakata {\em et al.}~\cite{Nakata_etal_1999} considered
the discrete-time two-party voter model on connected graphs, and
discussed its application to consensus  problems in distributed systems.
Both papers focus on
analysing the probability that all vertices
will eventually adopt the opinion which is initially held by a given group of vertices.

Let $A$ and $B$ denote also the sets of vertices
with opinions $A$ and $B$, respectively; $A \cup B= V$.
Let
$d(${\em X}$)$ be the sum of the degrees of the vertices in a set $X$.
We  say that opinion $A$ wins, if all vertices eventually adopt this opinion.
The central result of  \cite{HassinPeleg-InfComp2001} and \cite{Nakata_etal_1999} is that
the probability that opinion $A$ wins is 
\begin{equation}\label{PA}
P_A= \frac{d(A)}{2m},
\end{equation}
where $m$ is the number of edges in~$G$.
Thus in the case of connected regular graphs, the probability
that $A$ wins is proportional to the original size of $A$, irrespective of the graph structure.
%The result generalizes immediately to $k$-party voting with $A_i, i=1,...,k$ the set of vertices
%initially holding opinion $i$.
Apart from the probability of winning the vote, another quantity of interest is the time $T$
taken for voting to complete. In \cite{HassinPeleg-InfComp2001} it is  proven that $\E T =O(n^3 \log n)$ for general connected
$n$ vertex graphs.
For the case of random $d$-regular graphs, it is shown in~\cite{XXX} that
$\E T \sim 2n (d-1)/(d-2)$ with high probability.
%
%(for example),
%with high probability
%$\E T \sim 2n (d-1)/(d-2)$~\cite{XXX}.
%
It follows from the proof of this result
that, with high probability, two-party voting needs $\Th(n)$ time  to complete on
random $d$-regular graphs.

The performance of the two-party pull-voting seems unsatisfactory in  two  ways.
Firstly, it is
reasonable to require that a clear majority opinion wins
with high probability.
From \eqref{PA} it follows that, even if initially
only a single vertex $v$
holds opinion $A$, then this opinion  wins with probability $P_A=d(v)/2m$.
Secondly, the expected completion time is at least $\Om(n)$  on many classes of graphs,
including regular expanders and complete graphs.
This seems a long time to wait to resolve a dispute between two opinions. A more
reasonable waiting time would
depend on 
%be a slowly growing function of 
the graph diameter, which is $O(\log n)$ for many important
classes of 
graphs including
expanders.

To address these issues, we consider a modified version of pull voting in which
each vertex $v$ randomly queries two neighbours at each
step.
On the basis of the sample taken, vertex $v$ revises its opinion as follows.
If both neighbours have the same opinion, the calling vertex $v$  adopts this opinion.
If the two opinions differ, 
%of the selected neighbours are different, 
the calling vertex $v$ retains its current opinion in this round.
To distinguish this process from the conventional
pull voting, as described above, 
we use terms {\em single-sample voting}
and {\em two-sample voting}.
The aim of  the two-sample voting is to
ensure that voting finishes quickly and the initial majority opinion wins (almost always).
The two-sample voting is  intrinsically attractive, as it seems  to mirror the way people behave.
If you hear it twice it must be true.

In \cite{CEHR12} we  analysed a two sample process called min-voting.
Here, initially each vertex holds a distinct opinion. 
In each step every vertex
chooses two neighbours uniformly at random and takes the smaller opinion of the two.
For graphs with good
expansion properties we proved that min-voting completes in time $O(\log n)$,
with high probability.
%
%(\whp). 
%
Although min-voting is  fast, an
adversary with somewhat limited abilities could break the system by continuously
introducing small numbers into the network. Moreover the model is meaningless in two party voting, as the smaller 
opinion always wins.

In this paper we analyse two-sample voting for two classes of $d$-regular graphs: 
random graphs and expanders parameterized by the eigenvalue gap. 
Our results  depend only on the initial imbalance
$\nu_0 = (|A|-|B|)/n$. As an example, for  random  $d$-regular graphs 
there is an absolute constant $K$, independent of $d$, such that provided
\begin{equation}\nonumber
 \nu_0 \; \ge \; K \sqrt{\frac{d}{n} + \frac{1}{d}},
\end{equation}
with high probability two-sample voting is completed in $O(\log n)$ steps and
the winner is the opinion with the initial majority.
We discuss our results  in more detail in the next section. 
The main advantages
of our two-sample voting 
are that the completion time speeds up from $\Th(n)$ to $O(\log n)$
and with high probability
the initial majority opinion wins.

It seems interesting to enquire further how the performance of pull voting systems
depends on the range of choices available in the design.
%For simplicity 
We restrict our discussion to two-party voting.
%, so that the set of possible opinions is $\{0,1\}$.
The main issues seem to be
 the number of neighbours $k$ to contact
 at each step, and the rule used to reach a decision based on the
 opinions obtained.
In the case $k=1$, this is single-sample pull voting, as discussed above.
For $k=2$, a simple rule is to adopt the opinion if both neighbours agree
(the voting protocol analysed in this paper).
For $k \ge 3$ odd,
a comparable rule is to adopt the majority opinion.
Interestingly, the number $k$ of  neighbours
contacted at each step can substantially influence the performance of
the process in at least three ways: the completion time,
the final outcome, and
the robustness of the system against adversarial attacks.

We briefly compare the performance of such systems for two-party voting
on random $d$-regular graphs for various values of $k$.
Surprisingly, a clearly defined complexity hierarchy
emerges,
which distinguishes between $k=1,\; k=2$ and $k\ge5$ odd.

\begin{itemize}
\item $k=1$. As previously mentioned, the expected completion time 
in this case
is $\Th(n)$ with high probability. Let $A$ be the size of the initial majority opinion.
From \eqref{PA} we obtain that opinion  $A$ wins with probability $|A|/n$.
Thus if $|A|=cn$,  opinion $A$ wins with probability $c<1$, even if $A$ is a clear majority.

\item $k=2$. This is the topic of this paper.
We show that if the initial imbalance
between  the opinions  is not too small, then with high probability the time to completion is $\Th(\log n)$,
resulting in an exponential speed up over the case $k=1$,
and  the majority wins.
More details are given in the next section.

\item $k\ge 5$. It follows from the proof presented by Abdullah and Draief~\cite{AD}, that
for $k$ odd and
$d \ge k$
constant, if the initial
allocation of the opinions is chosen randomly,
the initial imbalance is sufficiently large, and the selection of $k$ neighbours is
done without replacement, then with high probability the  majority wins,
and the voting completes
in $\Th(\log \log n)$ rounds.
\end{itemize}

In the  particular case of the complete graph $K_n$, the performance of
two-party voting
is well studied. 
%Work includes the following.
Becchetti et al.~\cite{BCNP+13} consider the case $k=3$ in these graphs. 
The main focus of the work is on the completion 
time as a function of the number
of opinions. The result for two opinions is $O(\log n)$ provided the 
difference is not too small.
Cruise and Ganesh \cite{CG} consider a more general but asynchronous model.
Their work includes the case $k=2$,
and gives a $\Th(\log n)$ result.
A variant of
two-sample voting has been considered by Doerr et al.~\cite{DGMSS},
where the number of opinions can take any value from
$\{1, \dots , n\}$.
In their model, whenever a node $v$ contacts two neighbors $u$
and $w$, it adopts the median of the opinions of $u$, $v$ and $w$. Once the system is left
with two opinions, this protocol is equivalent to the two-sample voting considered in this paper.
If initially there are $s$ opinions, they showed
an $O(\log s \log \log n + \log n)$ convergence time to a
so-called ``stable consensus'' on complete graphs.
\ignore{
That is, if an adversary is able to change the opinion of $f$ vertices in each round,
where $f$ is at most $\sqrt{n}$, then within the given time bound all but $O(f)$ vertices
will agree on a common value, \whp\
Moreover, for $n^{\Th(1)}$ further rounds there
will continue to be $n-O(f)$ vertices with this common value \whp
%provided the distributed voting algorithm continues running.
}

\ignore{
In a recent paper, Becchetti et al.~\cite{BCNP+13} consider the case $k=3$ in
complete graphs.
That is, each vertex chooses three neighbors, and adopts the majority among these
three (breaking ties arbitrarily). They show that if the number of opinions is $s$ at the
beginning, then the voting process completes in time $O(\min\{s,(n/\log n)^{1/3}\} \log n)$ provided
that one of the votes dominates any of the other votes by a sufficiently large additive value,
which is ${\tilde{O}}(n^{2/3})$ if $s$ is not too large. Furthermore, they provide lower bounds
on the convergence time if $s \leq (n/\log n)^{1/4}$, as well as for the case
when $k \geq 3$ neighbors are consulted, where $s/k = O(n^{1/4-\epsilon})$ with
$\epsilon$ arbitrarily small.
}

An alternative approach to $k$-sample voting
is to use a majority dynamic.
In this case
each vertex adopts the most popular
opinion among {\em all\/} its neighbours.
In \cite{AD,MNT} the authors answered several important questions w.r.t.~majority voting
in general graphs, expanders, and random graphs, which we now describe.

Majority dynamics were studied by Mossel et al. in \cite{MNT}
who gave bounds for different
scenarios. They consider a model where initially
an opinion from $\{1, \dots , k\}$ is assigned to the vertices independently
according to a probability distribution. Then, the following deterministic
process is considered, which is fully defined by the initial distribution.
For $T$ time steps, each vertex
adopts the opinion held by the majority of its neighbors. After step $T$ a fair and monotone
election function is applied to the opinions of the vertices, resulting in a winning opinion.
Mossel et al.~\cite{MNT} showed that (under certain assumptions)
for the two-party model this process results in the correct (initial majority) answer.

\ignore{That is, if at the beginning the distribution is biased toward one of the opinions, then this
opinion wins.
The authors also analysed the question, how large is the population which
holds the same (correct) opinion after $T$ steps. Furthermore, they showed for certain expander graphs
and initial bias toward one opinion that all vertices will eventually adopt the same opinion with
high probability.
However, they provided no explicit time bound on the convergence time.
}

Recently,  majority voting was considered by Abdullah and Draief \cite{AD} on fixed degree sequence
random graphs. They
studied this process in the two-party case, where each vertex adopts the most popular opinion
among the neighbors in each step.
The initial opinions are distributed randomly according to a biassed
 distribution. They showed that if the initial bias toward one opinion is large enough, then with high probability
this opinion is adopted by all vertices within $O(\log \log n)$ time steps, and established
a similar lower bound.

\section{Our results for two-sample voting}\label{Us3}

\ignore{
 There are, however, four main different
aspects compared to our paper. Firstly, in our case, an adversary, which has full knowledge
of the network, is allowed
to place the opinions on the vertices at the beginning. Moreover,
the adversary is allowed to redistribute the opinions after each step. The only restriction
is to keep the number of different opinions unchanged.
Secondly, at each step every vertex only checks the
opinion of $k=2$ randomly chosen neighbours, which reduces communications complexity.
Thirdly, we show that this small change in the
ability of the vertices compared to standard ($k=1$) pull voting
is enough for an exponential improvement in the completion time of the
voting process, even in arbitrary graphs with second largest eigenvalue $\lambda_G < 3/5$
For the precise definition of $\lambda_G$ see below.
Fourthly, as is the case with \cite{DGMSS}, the system can maintain the majority opinion for
$n^{\Th(1)}$ steps even in the presence of $o(n)$ states persistently corrupted by an adversary.
}

Assume that initially each vertex holds
one of two opinions.
For convenience, $A$ and $B$ will denote the opinions,
the two sets of vertices who have these opinions, and the sizes of these sets,
depending on the context.
If opinion $A$ is  the majority, then the imbalance $\nu$ (also referred to as the relative difference
between the votes,
or the advantage of the $A$ vote)
is given by $A - B = \nu n$.
% (equivalently, $A= n(1+\nu)/2,\;B = n(1-\nu)/2$).

As mentioned earlier, for single-sample pull voting on  random $d$-regular graphs and expanders, 
the expected time  to complete is $\Th(n)$ with high probability\footnote{With high probability or \whp~means with probability tending to 1 as $n$ increases.}.
%
%(\whp).
%
%\ignore{
%$\E T = 2n (d-1)/(d-2)$.
%It follows that two-party voting almost always needs $\Th(n)$ time to complete
%irrespective of the degree $d$ of the underlying graph.}
Moreover, for any connected regular graph,
the probability that the initial majority $A$ wins the vote is only
$A/n$.
% irrespective of the graph structure.
In this paper we show that if 
there is
a sufficient initial imbalance between the two opinions, then with high probability
two-sample voting on $d$-regular random graphs
and expanders
is completed in a time which is logarithmic in the graph size, and
 the initial majority opinion wins.

Our results depend on the initial imbalance $\nu_0$
and, in the case of random regular graphs on 
the degree $d$ of the graph, while in the case of expanders 
%in the case of random regular graphs, 
on the second largest eigenvalue of the transition matrix.
%Our  results are for the class of random $d$-regular graphs, and for expanders
%parameterized by the eigenvalue gap.
A random $d$-regular graph is a graph sampled uniformly at random 
from the set of all $d$-regular graphs.
The results hold with high probability, 
which
depends on the selection of a graph (in the case of random graphs)
as well as on the voting process, which is itself probabilistic.

\begin{theorem}\label{Th1}
Let $G$ be a random $n$-vertex $d$-regular graph
with opinions $A$ and $B$ and  with initial imbalance $\nu_0 = |A-B|/n$.
There is an absolute constant $K$ (independent of $d$) such that, provided 
\begin{equation}\label{nu-nu-bound}
 \nu_0 \; \ge \; K \sqrt{\frac{d}{n} + \frac{1}{d}},
\end{equation}
with high probability two-sample voting is completed in $O(\log n)$ steps, and
the winner is the opinion with the initial majority.
\end{theorem}

\begin{corollary}\label{Corr1} (Sparse random graphs.)
Let $G$ be a random $d$-regular graph with $d \le \sqrt{n}$, and
let opinions $A$ and $B$ be placed on the vertices of $G$.
There is a constant $K$ such that, provided
\begin{equation}\label{nu-nu-bound-sparse}
 \nu_0 \; \ge \; \frac{K}{\sqrt{d}},
\end{equation}
with high probability two-sample voting is completed in $O(\log n)$ steps and
the winner is the opinion with the initial majority.
\end{corollary}

We give a similar result for expanders, that is, for a $d$-regular graph $G$ with a small
second eigenvalue $\lambda_G = \max \{\lambda_2, |\lambda_n|\}$, where $\lambda_1 \geq \lambda_2
\geq \dots \geq \lambda_n$ are the eigenvalues of the transition matrix $P = (1/d)A$ of a random walk on
$G$ and $A$ is the adjacency matrix of $G$.
%We assume that $\lambda_G < 3/5$.

\begin{theorem}\label{Th1-expanders}
Let $G$ be an $n$-vertex $d$-regular graph
and let opinions $A$ and $B$ be placed  on the vertices of $G$.
There is an absolute constant $K$ (independent of $d$ and  $\l_G$) such that, provided
\[ \nu_0 \; \ge \;  K \l_G,
%\min\left\{ K \l_G, \frac{6(3/5 - \lambda_G)}{13}\right\} ,
\]
 with high probability two-sample voting is completed in $O(\log n)$ steps and
the winner is the opinion with the initial majority.
\end{theorem}

Observe that for the above results to be non-trivial,  we should consider
$K^2 \le d \le n/K^2$ in Theorem~\ref{Th1},
$K^2 \le d \le \sqrt{n}$ in Corollary~\ref{Corr1},
and $\l_G \le 1/K$ in Theorem~\ref{Th1-expanders}.
According to the results above, in order to guarantee that two-sample voting (starting with small imbalance $\nu_0$) 
completes within $O(\log n)$ steps and the initial majority opinion wins,
it suffices to take a random $d$-regular graph with appropriately large degree $d$, or
an expander graph appropriately small $\l_G$.
We also show that the initial majority wins in $O(\log n)$ steps  
in random regular graphs with small degree as well as
in expanders with $\l_G$ not too small,
if the initial minority is a small constant fraction of the number of vertices.

\begin{theorem}\label{Thm-randgraph-small-d}
Let $ d > 10$ and let $G$ be a random $d$-regular $n$-vertex graph
with votes $A$ and $B$. 
There is a constant $c >0$ (independent of $d$) such that, provided
the initial size of the minority vote $B$ is at most $cn$,
with high probability two-sample voting is completed in $O(\log n)$ steps and
the winner is the initial majority opinion $A$.
\end{theorem}

\begin{theorem}\label{phaseIIandIII-expander}
Let $G$ be a $d$-regular $n$-vertex graph with
$\l = \l_G = 3/5 - \epsilon$.
Then, provided
the initial size of the minority vote $B$ is at most $(\epsilon/5)n$,
with high probability two-sample voting is completed in $O(\log n)$ steps and
the winner is the initial majority opinion~$A$.
\end{theorem}

The above theorems hold under following adversarial conditions.
The adversary has full knowledge of the graph, decides the initial distribution of the opinions
among the vertices, and can arbitrarily redistribute the opinions at the beginning of each voting step.
The adversary cannot change the number of opinions
of each type. 
However, one can trace our proofs to see that
if we allow the adversary to change the opinions of at most
$f=o(\nu_0 n)$ vertices during the execution of the algorithm, then under the conditions considered in 
Theorems~\ref{Th1} and~\ref{Th1-expanders} and Corollary~\ref{Corr1}, after $O(\log n)$ voting steps
all but $O(f)$ vertices adopt the majority opinion $A$. That is, with high probability the protocol
can surpress persistently corrupted vertices (cf.~\cite{DGMSS}).

As described in the previous section, it
seems that there is a clearly defined hierarchy w.r.t.~distributed voting in
random regular graphs.
If every node is only allowed to consult
one single neighbour (and adopt its opinion), then -- as shown in \cite{CEHR-SIAM2013} -- one requires
$\Theta(n)$ steps on the graphs we consider
to converge to one opinion. If every node can consult two neighbours
(selecting them randomly with or without replacement) and adopt
the opinion of these two vertices if they are the same, then the running time is $O(\log n)$, so exponentially
faster. On the other side, 
even if the adversary is not allowed to re-distribute votes,
$\Omega(\log_d n)$ is a natural lower bound in any
$d$-regular graph.
This holds since there might be initially $\Theta(n)$ pairs of adjacent
vertices 
all with the same minority opinion $B$. The vertices of such a pair choose each other with
probability $\Theta(1/d^2)$, in which case none of them will change its opinion. Thus, the protocol
needs $\Omega(\log_d n)$ steps in order to guarantee that in none of these $\Theta(n)$
pairs the vertices choose each other all the time.
This lower bound holds also for $k$-sample voting for a constant $k\ge 3$,
if the selection of $k$ neighbours is done with replacement
(if a $B$ vertex $v$ has a $B$ neighbour, then $v$ does not change its opinion in the current step
with probability $\Omega(1/d^k)$). 
If every node may contact at least five different neighbours 
(selection without replacement) and adopt the majority opinion among them, then
on random regular graphs with randomly distributed opinions (biased toward $A$), $\Theta(\log \log n)$
steps suffice until $A$ wins 
(this follows from the analysis in~\cite{AD}).

We should mention that the last result does not hold if the opinions are not randomly distributed.
An adversary could assign the minority opinion $B$ to a vertex $v$ as well as all vertices 
which are at distance at most
$Diam/3$ to $v$, where $Diam$ denotes the diameter of the graph.
Clearly, the voting protocol needs at least $Diam/3$ steps.
Also, the example in the next paragraph shows that the result w.r.t.~$5$-sample voting \cite{AD}
cannot be extended to graphs with similar conductance as in random graphs.

Consider a random regular graph of degree $d-5$ with $n$ 
vertices, where $d = \omega(1)$. 
Then, we group the vertices in clusters of size $6$ -- we assume that 
$n$ is a multiple of $6$. In each cluster, we connect every vertex with all the 
other $5$ vertices, so the graph has degree $d$.
Bollob\'as \cite{Bol88} shows that 
the conductance of this graph is with high probability at least $1/2 - o(1)$ and at most 
$1/2(1+o(1)) + 5/d$. The upper bound holds since for any subset 
$S$ with $|S| \leq n/2$, there are at least $(1-o(1))(d-5) |S|/2$ edges crossing
the cut between $S$ and $V \setminus S$ in a $d-5$-regular random graph \cite{Bol88}. 
On the 
other hand, there is a set for which the size of the cut is at most 
$(1+o(1))(d-5) |S|/2$.
Furthermore, 
each node has $5$ additional (inner-cluster) edges, which may increase the cut. 
Since $d=\omega(1)$, the conductance of this graph is with high probability $1/2 \pm o(1)$, which is 
almost the same as in a $d$-regular random graph.
Clearly, all vertices from a cluster may
choose each other in a step with some probability larger than $1/d^{30}$. Thus, the protocol
needs $\Omega(\log_d n)$ steps in order to guarantee that in none of these clusters the vertices choose
each other all the time. 

Concerning our results, the constant eigenvalue gap $1-\l_G$ (as in Theorem~\ref{phaseIIandIII-expander})
seems to be needed. 
For example, consider a hypercube with $d = \log n$ and $1 - \l_G = o(1)$.  
If the adversary is allowed to rearrange the opinions in
each step, then we may have for $\Omega(d^2)$ steps configurations, in which all vertices of a subcube
of dimension $d-c$ have opinion $B$, where $c$ is a constant. 
Such a $B$-vertex converts to $A$ with probability
$(c/d)^2$, so $\Omega(d^2) = \Omega(\log^2 n)$ steps are needed for the protocol to finish.

The proof techniques used for single-sample voting, namely coalescing random walks,
do not apply in the case of two-sample voting. Our proofs are based on concentration of the size of edge cuts around the expectation
coupled with a worst case analysis.
%For $d$ small constant the size of  edge cuts is not very concentrated,
%but as $d \rai$ the concentration improves
%and we are able to prove strong results.

\section{Background material and outline of proof}

The analysis of the voting process is made in the following  three phases, 
where $B$ is the minority vote.
%\begin{description}

\noindent
\makebox[5.5em][l]{\bf Phase I:} $cn \le B \le n(1-\nu_0)/2$.\\
\makebox[5.5em][l]{\bf Phase II:} $\om \le B \le cn$.\\
\makebox[5.5em][l]{\bf Phase III:} $1 \le B \le \om$.
%\end{description}

Let $B(t)$ denote the set of vertices with opinion $B$
and the size of this set in step $t$.
Whenever it is
clear from the context, we write $B$ instead of $B(t)$.
Phase I reduces $B(t)$ from
$B(0)= n(1-\nu_0)/2$ to $B(T) \le cn$, for some small constant $c$, 
in a sequence of $T=O(\log (1/\nu_0))$ rounds.
%The objectives are to prove that \whp\ $B(t+1)$ is smaller than $B(t)$, 
%and to establish a bound on the time to complete Phase I.
%
%Problems arise both in proving
%that \whp\ $B(t+1)$ is smaller than $B(t)$, and in establishing the bound on the time to complete Phase I.
%
The reduction of
$B(t)$ in Phase II is more dramatic.
The $\omega$ threshold between phases II and III is a function slowly growing with~$n$.
In Phase III things may slow 
down
again and
the last few steps 
can be viewed as
%
%are little more than 
%
a biassed random walk.
All phases are analysed
in the adversarial model, which allows an arbitrary redistribution of the votes at the beginning of each step.

The following Chernoff--Hoeffding inequalities are used throughout the proofs.
%For a proof of \eqref{Cher1},  see e.g. Alon and Spencer \cite{AS}.
Let   $Z=Z_1+Z_2+\cdots Z_N$ be the sum of  the independent random variables
$0\leq Z_i\leq 1,\,i=1,2,\ldots,N$, $\E(Z_1+Z_2+\cdots+Z_N)=N\m$, and $0 \le \e \le 1$. Then
\begin{eqnarray}
\Pr(Z \leq (1- \e )N\m) & \leq & e^{-\e^2N\m/3},  \label{Cher1_new} \\
%\Pr(Z\geq \a N\m)&\leq&(e/\a)^{\a N\m}.\label{Cher2}
\Pr(Z \geq (1+ \e )N\m )&\leq& e^{-\e^2N\m/2}.\label{Cher1}
%\Pr(Z\geq \a N\m)&\leq&(e/\a)^{\a N\m}.\label{Cher2}
\end{eqnarray}
For any $\epsilon >0$, we have
\begin{eqnarray}
%\Pr(Z \leq (1- \e )N\m) & \leq & e^{-\e^2N\m/3},  \label{Cher1_new} \\
%\Pr(Z\geq \a N\m)&\leq&(e/\a)^{\a N\m}.\label{Cher2}
\Pr(Z \geq (1+ \e )N\m )&\leq& \left(\frac{e^{\epsilon}}{(1+\epsilon)^{1+\epsilon}}\right)^{N\m}.\label{Cher2}
%\Pr(Z\geq \a N\m)&\leq&(e/\a)^{\a N\m}.\label{Cher2}
\end{eqnarray}
%We refer to \eqref{Cher1} and \eqref{Cher2} as the Hoeffding and Chernoff Inequalities (respectively).

Our proofs for the case of random graphs
are made using the configuration model of $d$-regular $n$-vertex multigraphs.
Let $\cC_{n,d}$ be the space of $d$-regular $n$-vertex configurations,
and let $\cC^*_{n,d}$ be the sub-space of $\cC_{n,d}$ of the configurations
whose underlying
graphs are simple.
A configuration $S$ is a matching of the $nd$ ``configuration points''
(each vertex is represented by $d$ points).
Every simple graph maps to the same number of configurations, so $\cC^*_{n,d}$ maps
uniformly onto ${\cal G}_{n,d}$, the space of $d$-regular $n$-vertex graphs.
We use the following result of \cite{FL} for
the size of $|\cC^*|/|\cC|$. See e.g. \cite{CF} for a proof.

\begin{lemma}
Let $1 \le d \le n/8$. If $S$ is chosen uniformly at random from $\cC_{n,d}$,
then
\begin{equation}\label{simple}
\Pr(S \in \cC^*_{n,d}) \ge e^{-20d^2}.
\end{equation}
\end{lemma}

This lemma is used in the following way.
Let $\ul Q$ be a property of $d$-regular $n$ vertex multigraphs.
Then, denoting by $G(S)$ the underlying multigraph of configuration $S$, 
\begin{eqnarray} 
  \Pr_{\cG} (G\in \ul Q) 
    & = & \Pr_{\cC} (G(S)\in \ul Q \: | \: S\in \cC^*)
    \; \le \; \frac{  \Pr_{\cC} (G(S)\in \ul Q) }{\Pr_{\cC}(S\in \cC^*)} \nonumber \\
    & \le & \Pr_{\cC} (G(S)\in \ul Q) \cdot e^{20d^2}. \label{njcwnc82s}
\end{eqnarray}

At any step $t$ of the voting process, let $\D_{AB}=\D_{AB}(t)$
be the number of $A$ vertices converting to $B$ during this step. Similarly, let $ \D_{BA}$
be the number of $B$ vertices converting to $A$ during step $t$. 
At each step we obtain a lower bound on
$\E  \D_{BA}$, an upper bound on $ \E  \D_{AB}$, and use the concentration
of these two random variables given by~\eqref{Cher1_new} and~\eqref{Cher1}
to get a \whp\ value of $ \D=\D_{BA}-\D_{AB}$, which is the increase of the number of
$A$ vertices in this step.

For a vertex $v$ and a set of vertices $C$,
let $d_v^C$ be the number of vertices in $C$ which are adjacent to $v$.
For $v \in A$, let $X_v=1$ if $v$ chooses twice in $B$ at step $t$, and 0 otherwise. Thus
\[
\D_{AB}= X_A=\sum_{v \in A} X_v
\]
The $X_v$ are independent $\set{0,1}$ random variables with the expected value
depending whether the neighbours are selected with or without replacement: 
\[
\E X_v(\text{with replacement}) = \bfrac{d_v^B}{d}^2, \;\;\; 
%\mbox{and} \;\;\;\;\;
\E X_v(\text{no replacement})  = \; \frac{(d_v^B)( d_v^B-1)}{d(d-1)}.
\]

We give proofs for sampling with replacement.
The proofs for sampling without replacement  follow because
\[
\E X_v(\text{no replacement}) = \frac{d}{d-1} \E X_v(\text{with replacement}),
\]
so all inequalities for expected values in one model imply similar inequalities in the other model.

%The Phase III proof for sampling with replacement requires more care.

%%%%%%%%%%%%%%%%%%%%%%%%%%  PHASE I STARTS HERe
%

\section{Phase I of analysis: $cn \le B \le n(1-\nu_0)/2$}\label{Sec-Phase-I}

\subsection{The main lemma and its applications to expanders}

Lemma~\ref{lemMethod2-new} below gives a sufficient condition for a fast reduction of the 
minority $B$-vote
from $(1-\nu_0)n/2$ to $cn$, where $\nu_0 < 1$ and $c< (1-\nu_0)/2$.
For example, for $\nu_0 = 1/10$ and $c = 1/20$,
the $B$ vote reduces from $(9/20)n$ to $(1/20)n$. 
The condition in Lemma~\ref{lemMethod2-new}
says that the number $E(X,Y)$ of edges between any disjoint large subsets of vertices
$X$ and $Y$ is close to the value $dXY/n$ expected in the random regular graph.
This condition is of the form
as in the Expander Mixing Lemma (stated below as Lemma~\ref{bcwjhcbw3}),
so Lemma~\ref{lemMethod2-new}
can be immediately applied to expanders (see Corollary~\ref{lemMethod2-expanders}).
Lemma~\ref{lemMethod2-new} can also be applied without a reference to the second eigenvalue
(if the second eigenvalue is not known or is not good enough)
by directly checking that large subsets of vertices are connected by many edges.
We illustrate this by considering random $d$-regular graphs for any $d \in [K, n/K]$, where
$K$ is some (large) constant (see Lemma~\ref{lem-nuBC} and Corollary~\ref{lemMethod2-randomgraphs}).

\begin{lemma}\label{lemMethod2-new}
Let $0 < c \le 1/2$, $0 < \a \le  c^{3/2}/36$, and $\a^2 c^2 n = \Omega(n^\e)$, for 
a constant $\e > 0$.
Let $G$ be a $d$-regular $n$-vertex connected graph such that
\begin{equation} \label{mixing-prop}
   \left| E(X,Y) - \frac{dXY}{n} \right| \le \a d \sqrt{XY},
\end{equation}
for each pair $X$ and $Y$ of disjoint subsets of vertices of
sizes $Y \ge cn$ and $X \ge (2/3)\a c^{3/2} n$.
There exist absolute constants $K$ and $K'$ (independent of $d$, $c$ and $\a$) such that,
if the initial advantage of the $A$-vote in $G$ is
\begin{equation}\label{nu0-bound-x}
 \nu_0 \; \ge \; K \a,
\end{equation}
then with probability at least $1 - e ^{-\Theta(\a^2 c^2 n)}$, the advantage of the $A$-vote
increases to $1-2c$ (that is, the $B$-vote decreases to $cn$)
within $K'(\log (1/\nu_0) + \log( 1/c))$ voting steps.
\end{lemma}

The parameters $c$ and $\a$ in the above lemma can be considered as some small constants, but they can also
depend on $d$ (and decrease with increasing $d$).

%\subsection{Application to expanders}

\begin{lemma}\label{bcwjhcbw3} {\rm (Expander Mixing Lemma~\cite{ExpanderLemma}).}
Let $G = (V,E)$ be a $d$-regular $n$-vertex graph.
Let $1 = \l_1 \ge \l_2 \ge \cdots \l_n \ge -1$ be the eigenvalues
of the transition matrix of the random walk on $G$,
and let $\l = \l_G = \max\{ |\l_2|, |\l_n| \}$. Then for all $S, T \subseteq V$,
\[ \left| E(S,T) - \frac{dST}{n} \right| \; \le \; \l d \sqrt{ST}.
\]
\end{lemma}

Lemmas~\ref{lemMethod2-new} and~\ref{bcwjhcbw3} imply the following corollary.

\begin{corollary}\label{lemMethod2-expanders}
For any constant $0 < c < 1/2$,
there exist constants $K_1$ and $K_2$ (which depend on $c$) such that
for any regular $n$-vertex graph $G$
with the initial advantage of the $A$-vote
$\nu_0 \; \ge \; K_1 \l$,
the minority vote $B$ decreases to $cn$
within $K_2\log (1/\nu_0))$ voting steps,
with probability at least $1 - e ^{-\Theta(\l^2 n)}$.
%where $\l = \l_G$.
\end{corollary}

\begin{proof}
Let $K_1 = \max\{K, 36/c^{3/2}\}$, where constant $K$ is from Lemma~\ref{lemMethod2-new}.
If $\l \ge 1/K_1$, then the statement of the corollary is trivially fulfilled.
If $\l \le 1/K_1$, then $\l \le c^{3/2}/36$ and we can apply Lemma~\ref{lemMethod2-new}
with $c$ and $\a = \l$.
Now Lemmas~\ref{lemMethod2-new} and~\ref{bcwjhcbw3}
imply that if $\nu_0 \ge K_1\l$, then 
with probability at least $1 - e ^{-\Theta(\l^2 n)}$, the size of the $B$-vote decreases 
to $cn$ within $K'(\log (1/\nu_0) + \log( 1/c)) = K_2(\log (1/\nu_0)$ voting steps.
\end{proof}

\subsection{Application to random regular graphs}

If $d=O(1)$, then a random $d$-regular graph has $\lambda_G \leq
(2\sqrt{d-1} + \epsilon)/d$, \whp, where $\epsilon>0$ can be any small constant
\cite{Fri03}. 
Thus,
for $d =O(1)$ 
Corollary~\ref{lemMethod2-expanders} applies.
To apply Lemma~\ref{lemMethod2-new}
to 
random regular graphs with degree 
which may grow with the number of vertices, we need to establish a suitable $\a$ for 
the bound in~(\ref{mixing-prop}) without refering to $\l_G$.
%larger than $600$.
The bound we show in the next lemma is stronger than a similar bound 
shown by Fountoulakis and Panagiotou~\cite{FountoulakisPanagiotou}.
Using the bound from~\cite{FountoulakisPanagiotou} would lead to a weaker 
relation between $\nu_0$ and $d$ than in Theorem~\ref{Th1}.

\begin{lemma}\label{lem-nuBC}
For given set sizes $X \le Y$, 
in a random $d$-regular $n$-vertex graph $G = (E,V)$,
with probability at least $1 - 2e^{-Y}$,
each pair of disjoint subsets of vertices $\cX$ and $\cY$ of sizes $X$ and~$Y$
satisfies the following inequality:
\begin{equation} \label{mixing-prop-randgraphs}
   \left| E(\cX,\cY) - \frac{dXY}{n} \right| \le d \sqrt{XY}
                     \sqrt{\frac{1}{d}\, 24\,\log(ne/Y) + \frac{d}{Y} 160}.
\end{equation}
\end{lemma}

\begin{corollary}\label{lemMethod2-randomgraphs}
For any constant $0< c < 1/2$,
there exist constants $K_1$ and $K_2$ (which depend on $c$) such that
for a random $d$-regular $n$-vertex graph with
the initial advantage of the $A$-vote
\begin{equation}\label{nu0-bound}
 \nu_0 \; \ge \; K_1\sqrt{\frac{1}{d} + \frac d n},
\end{equation}
the minority vote $B$
decreases within $K_2\log (1/\nu_0)$ steps to $cn$,
with probability at least $1 - e^{-\Theta(n^{1/2})}$.
\end{corollary}

\begin{proof}
Let $0< c < 1/2$ be a constant and let 
\begin{equation}\label{jenq} 
   \a = \sqrt{\frac{1}{d} 24\log(e/c) + \frac{d}{n} (160/c)}
       = \Theta\brac{\sqrt{\frac{1}{d} + \frac d n}}.
\end{equation}
Lemma~\ref{lem-nuBC} implies that
for a random $d$-regular $n$-vertex graph,
the probability that Inequality~\eqref{mixing-prop} holds 
for each pair $X$ and $Y$ of disjoint subsets of vertices
such that $Y \ge cn$
is at least $1 - 2n^2e^{-cn} = 1 - e^{\Theta(n)}$.
If Inequality~\eqref{mixing-prop} holds for all such pairs of subsets of vertices, then
Lemma~\ref{lemMethod2-new} and~\eqref{jenq} imply that 
there are constants $K_1$ and $K_2$ such that 
if the initial vote imbalance is  $\nu_0 \ge K\a \ge K_1\sqrt{1/d + d/n}$, then with probability at least
$1 - e^{\Theta(\a^2 n)} \ge 1 -e^{\Theta(n^{1/2})}$,
the minority vote $B$ decreases to $cn$ in $K_2\log (1/\nu_0)$ steps. 
%(but not the same as in Lemma~\ref{lemMethod2-new}).
Thus with probability at least
$(1 - e^{\Theta(n)})(1 -e^{\Theta(n^{1/2})}) = 1 -e^{\Theta(n^{1/2})}$,
for a random $d$-regular $n$-vertex graph
with the initial advantage of the $A$-vote $\nu_0 \ge K_1\sqrt{1/d + d/n}$,
the minority vote reduces to $cn$ in $K_2\log (1/\nu_0)$ steps.
\end{proof}

\subsubsection*{Proof of Lemma~\ref{lem-nuBC}}

Let $X,Y$ be two fixed disjoint vertex sets with sizes $X \le Y$
($X$ and $Y$ stand for the sets and their sizes).
Let $Z(S)$ be the number of edges between $X$ and $Y$ in a configuration
$S\in \cC_{n,d}$.
We order the $nd$ configuration points so that the first $dX$ points correspond to the vertices in set $X$.
A configuration $S$ can be represented 
as a sequence $(t_1,t_2, \ldots\ t_q)$, where $q = nd/2$, $1 \le t_i \le q - (2i-1)$, and
the number $t_i$ defines the $i$-th (matched) pair in $S$,
assuming the lexicographic order of pairs.
Denoting by $L_i$ the sequence of the remaining unmatched points
after the first $(i-1)$ pairs have been selected,
the $i$-th pair matches the first point in $L_i$ (which becomes the first point of this pair)
with the point in $L_i$ at the position $1+t_i$.
A random $S\in \cC_{n,d}$ is determined by independent random selections of $t_i$'s.
By considering first the points corresponding to the vertices in $X$, we ensure that $Z(S)$
is determined by $(t_1, t_2, \ldots, t_{dX})$.

For a configuration $S = (t_1,t_2, \ldots\ t_q)$ and $i = 0, 1, \ldots, q$, let
\begin{eqnarray*}
  Z_i(S) & \equiv & Z_i(t_1,t_2,\ldots, t_q)
      \; = \; \E_{\t_{i+1},\ldots,\t_{q}} Z(t_1,\ldots,t_i,\t_{i+1},\ldots,\t_{q})
      \; \equiv \; Z(t_1,t_2,\ldots,t_i).
\end{eqnarray*}
That is, $Z_i(S)$ is the expected number of edges between the $X$ and $Y$ points
in a random configuration which agrees with the configuration $S$ on the first $i$ pairs.

We have $Z_0(S) = \E Z(S) = dXY/n$ and
$Z_{dX}(S) = Z(S)$.
The sequence of random variables $Z_i$, $i = 0, 1,\ldots, q$ is a martingale
because $\E(Z_{i+1} | Z_i) = Z_i$:
\begin{eqnarray}
  \E(Z_{i+1} | Z_i = z) & = & \E(Z(t_1,\ldots,t_{i+1}) | Z(t_1,\ldots,t_{i}) = z) \nonumber \\
   & = &
   \sum_{t_1,\ldots,t_{i+1}: Z(t_1,\ldots,t_{i}) = z}
          \frac{\Pr(t_1,\ldots,t_{i+1})}{\Pr(Z_i = z)}  \; Z(t_1,\ldots,t_{i+1})  \nonumber \\
   & = & \frac{1}{\Pr(Z_i = z)}  \sum_{t_1,\ldots,t_{i}: Z(t_1,\ldots,t_{i}) = z}
            \Pr(t_1,\ldots,t_{i})  \; \sum_{t_{i+1}}  \Pr(t_{i+1}) Z(t_1,\ldots,t_{i+1})  \nonumber \\
   & = & \frac{1}{\Pr(Z_i = z)}  \sum_{t_1,\ldots,t_{i}: Z(t_1,\ldots,t_{i}) = z}
            \Pr(t_1,\ldots,t_{i}) \: Z(t_1,\ldots,t_{i}) \nonumber \\
   & = & z. \nonumber
\end{eqnarray}

Let $F^X_i(S)$ and $F^Y_i(S)$ denote the number of available (unmatched)
$X$-points and $Y$-points, respectively,
after the first $i$ pairs in $S$ have been matched. 
If $F^X_i(S) = 0$, then $Z_j(S) = Z(S)$, for all $j \ge i$.
If $F^X_i(S) \ge 1$, then, dropping $S$ from the notation for
simplicity,
\[ Z_i \; = \; (Y - F^Y_i) + \frac{F^X_i F^Y_i}{nd - (2i+1)}.
\]
The first term $Y - F^Y_i$ is the number of edges between the $X$ and $Y$ points given 
by the first $i$ pairs in $S$, 
and the second term is the expected number of edges between the $X$ and $Y$ points
contributed by a random matching of the remaining points.
(Each available point $x$ in $X$ contributes $F^Y_i/(nd - (2i+1))$ to this expectation,
because this is the probability that $x$ is matched with a point in $Y$.)
When the next $(i+1)$-st pair is matched, then $(F^X_{i+1}, F^Y_{i+1})$ is
either $(F^X_{i} - 1, F^Y_{i}-1)$,
or $(F^X_{i} - 2, F^Y_{i})$, or $(F^X_{i} - 1, F^Y_{i})$, and it can be
checked that in all three cases
\begin{equation}\label{bchjqwcq56}
 |Z_{i+1} - Z_i| \; \le \; 2.
\end{equation}

Alternatively, the bound~\eqref{bchjqwcq56} can be established by applying
the general {\em switching method}
(see~\cite{Wormald-1999} for discussion of this method).
If two configurations $S'$ and $S''$ differ only by two pairs,
that is, $S''$ can be obtained from $S'$ by ``switching'' two pairs $(a,b)$ and $(c,d)$, where $a < c$, to
pairs $(a,c)$ and $\{b,d\}$, or
$(a,d)$ and $\{b,c\}$, then clearly
\begin{equation}\label{bcwhcb901}
 |Z(S') - Z(S'')| \; \le \; 2.
\end{equation}
(Note that we use the set notation for the pairs $\{b,d\}$ and $\{b,c\}$, because
we do not know the relative order of points $b$ and $d$, and $b$ and $c$.)

Define $\cC_i(S)$ as the set of all configurations $S'$ which have the same first $i$ pairs as in $S$.
For a configuration $S' = ((a'_1,a'_2), \ldots, (a'_{nd-1},a'_{nd}))$,
$1 \le i < nd/2$ and $2i \le j \le nd$,
let $S'[i,j]$ be the configuration obtained from $S'$
by switching the $i$-th pair $(a'_{2i - 1}, a'_{2i})$ and the pair $\{a'_j,b\}$
for the pairs $(a'_{2i - 1}, a'_{j})$ and $\{a'_{2i},b\}$
(note that $S'[i,2i] = S'$).
The set $\cC_i(S)$ can be obtained from the set $\cC_{i+1}(S)$ by
replacing each configuration $S' \in \cC_{i+1}(S)$
with the configurations $S'[i+1, j]$, for $j = 2i+2, \ldots, nd$:
\begin{equation}\label{nvjev9012}
  \cC_i(S) \; = \; \{ S'[i+1, j]:\: S'\in \cC_{i+1}(S), \: 2i+2 \le j \le nd \}.
\end{equation}
Using~\eqref{nvjev9012},
we can write $Z_{i+1}(S) - Z_i(S)$ as
\begin{eqnarray}
  Z_{i+1}(S) - Z_i(S) & = &
    \frac{1}{|\cC_{i+1}(S)|} \sum_{S'\in \cC_{i+1}(S)} Z(S')
    \; - \; \frac{1}{|\cC_{i}(S)|} \sum_{S'\in \cC_{i}(S)} Z(S') \nonumber \\
    & = &  \frac{nd-(2i+1)}{|\cC_{i}(S)|} \sum_{S'\in \cC_{i+1}(S)} Z(S')
    \; - \; \frac{1}{|\cC_{i}(S)|}
              \sum_{S'\in \cC_{i+1}(S)} \: \sum_{j= 2i+2}^{nd} Z(S'[i+1,j])  \nonumber \\
    & = &   \frac{1}{|\cC_{i}(S)|}
            \sum_{j= 2i+2}^{nd} \: \sum_{S'\in \cC_{i+1}(S)} \left( Z(S') - Z(S'[i+1,j])\right).  \nonumber
\end{eqnarray}
There are $|\cC_{i}(S)|$ terms in the last double sum and the absolute value of each term 
is at most 2
(Inequality~\eqref{bcwhcb901}), so the bound~\eqref{bchjqwcq56}
follows.

The Azuma-Hoeffding inequality says that
if a sequence of random variables $(X_0, X_1, \ldots, X_N)$ is a martingale
and $|X_{i+1} - X_i| \le c$, for each $1 \le i \le N-1$, then
for any $\d$,
 \[
 \Pr( |X_N - X_0| \ge \d ) \; \le \; 2\exp\brac{- \frac{\d^2}{2Nc^2}}.
 \]
Applying this inequality to our martingale $(Z_0, Z_1, \ldots)$, we get
\[
\Pr_{\cC}\brac{ \left|E(X,Y) - \frac{dXY}{n}\right| \ge \d } \; = \;
\Pr( |Z_{dX}(S) - Z_0(S)| \ge \d ) \; \le \;
2\exp\brac{- \frac{\d^2}{8dX}}.
\]
Thus for a random $d$-regular $n$-vertex graph $G = (E,V)$, using~\eqref{njcwnc82s},
\[
\Pr_{\cG}\brac{ \left|E(X,Y) - \frac{dXY}{n}\right| \ge \d } \; \le \;
     2\exp\brac{- \frac{\d^2}{8dX} + 20d^2}.
\]

The number of pairs of disjoint sets of sizes $X \le Y$ is at most
\[
{n \choose X}{n -X \choose Y}
\; \le \; \brac{\frac{ne}{X}}^X  \brac{\frac{ne}{Y}}^Y
\; \le \; \brac{\frac{ne}{Y}}^{2Y}.
\]
The last inequality holds because $(ne/z)^z$ is monotone increasing for $0 \le z \le n$.
Therefore, using the union bound, for given set sizes $X \le Y$
and a random $d$-regular $n$-vertex graph $G = (E,V)$,
\begin{eqnarray}
\lefteqn{\Pr( \mbox{there are disjoint $\cX, \cY\subseteq V$
   of sizes $X$ and $Y$ such that}\; |E(\cX,\cY) - dXY/n| \ge \d )} \nonumber \\
   & &  \;\;\;\;\;\;\;\;\;\;\;\;\;\;\;\; \le \;
     2 \exp\brac{- \frac{\d^2}{8dX} + 20d^2 + 2Y\log(ne/Y)}.
          \;\;\;\;\;\;\;\;\;\;\;\;\;\;\;\;\;\;\;\;\;\;\;\;\;\;\;\;\;\;\;\;\;\;\;\;\;\;\;\;\;  \nonumber
\end{eqnarray}
The above bound is at most $2e^{-Y}$, if
\[
    \frac{\d^2}{8dX} - 20d^2 - 2Y\log(ne/Y) \; \ge \; Y,
\]
which is equivalent to
\begin{equation}\label{bnkqnds33}
 \d \ge d\sqrt{XY} \sqrt{\frac{1}{d}(16\log(ne/Y) + 8) + \frac{d}{Y}160}.
\end{equation}
Since the right-hand side in~\eqref{mixing-prop-randgraphs} is at least the 
right-hand side in~\eqref{bnkqnds33},
we conclude that~\eqref{mixing-prop-randgraphs} holds for all pairs of disjoint 
sets $\cX$ and $\cY$ of sizes $X \le Y$ 
with probability
at least $1 - 2e^{-Y}$.
%\end{proof}

\subsection{\bf Proof of Lemma~\ref{lemMethod2-new}}

Let $c$ and $\a$ be as in the statement of the lemma.
Let $A$ and $B$ be the voting groups at the beginning of the current step,
let $\nu = (A-B)/n$, and assume that for some sufficiently large constant $K$,
\begin{equation}\label{nuRange-Appx}
K \a \; \le \; \nu \; \le \; 1-2c,
\end{equation}
or equivalently,
%\begin{equation}\label{nuRange-Appx}
\[
        cn \; \le \; B \; \le \; \frac{1 - K \a}{2}n.
\]
%\end{equation}

{\bf Lower bound on $\D_{BA}$.}
The expectation $\E \D_{BA}$, taken over the random selection of neighbours in the current voting step,
is equal to
\begin{eqnarray}
\E \D_{BA} & = & \sum_{v \in B} \bfrac{d_v^A}{d}^2
  \; \ge \; \frac{1}{Bd^2} \brac{\sum_{v \in B} d_v^A}^2
  \; = \; \frac{1}{Bd^2}  \brac{E(A,B)}^2  \nonumber \\
  & \ge & \frac{1}{Bd^2}  \brac{ \frac{dAB}{n} - \a d \sqrt{AB} }^2
  \; = \; \frac{A^2B}{n^2} \brac{1 - \a \frac{n}{\sqrt{AB}}}^2  \label{hjre21} \\
  & \ge & \frac{A^2B}{n^2} \brac{1 - 2\a \frac{n}{\sqrt{AB}}}
  \; = \; \frac{A^2B}{n^2} \brac{1 - 2\eta}, \label{be56sw-Appx}
\end{eqnarray}
where $\eta \equiv \eta_{AB} = \a n/\sqrt{AB} \le \a \sqrt{2/c} \le c/25$.
The inequality on line~\eqref{hjre21} is the inequality~\eqref{mixing-prop} applied to sets $A$ and $B$
(as $X$ and $Y$).
Therefore,
\begin{eqnarray}
\Pr\brac{ \D_{BA} \le  \frac{A^2B}{n^2}(1-3\eta)}
   & \le & \Pr\brac{ \D_{BA} \le (1-\eta) \E\D_{BA}}
   \; \le \;      e^{-\eta^2 c n/24}
   \; \le \;      e^{-\a^2 c n/6}.   \label{eq-fger}
\end{eqnarray}
The second inequality above follows from~\eqref{Cher1_new}
applied with $\ep = \eta$, and the fact that
$\E\D_{BA} \ge (1/8)cn$ (from~\eqref{be56sw-Appx}).
The last inequality in~\eqref{eq-fger} holds because 
the definition of $\eta$ implies
\begin{equation}\label{hjcbw671}
 2\a \: \le \eta \le \frac{\a}{\sqrt{(1-c)c}}.
\end{equation}

{\bf Upper bound on $\D_{AB}$.}
Let $q = \lfloor (1/2)\log_2(n^2/(\eta B^2)\rfloor +1$.
We partition set $A$ into $q+1$ sets $C_i$, according to the
values $d_v^B$:

\begin{eqnarray*}
C_0 & = &  \set{v \in A: d_v^B  < (1+\eta) \frac{dB}{n}}; \\
C_i  & = &   \set{v \in A: (1+2^{i-1}\eta) \frac{dB}{n} \le d_v^B < (1+2^i\eta) \frac{dB}{n}},
      \;\;\;\;\; \mbox{for $i = 1,2,\ldots,q-1$};\\
C_q & = &   \set{v \in A: (1+2^{q-1}\eta) \frac{dB}{n} \le d_v^B}.
\end{eqnarray*}
Obviously $C_0 \le A$, and we show that $C_i \le A/2^{2(i-1)}$, for all $1 \le i \le q$.
For each $1 \le i \le q$, we have
\begin{eqnarray}\label{bchwbc3}
E(C_i,B) & = & \sum_{v\in C_i} d_v^B \; \ge \; \brac{1+2^{i-1}\eta}\frac{dC_iB}{n}.
\end{eqnarray}
If we had $C_i > A/2^{2(i-1)}$ for some $1 \le i \le q$, then,
from the definition of $q$ and $\eta$,
\[ C_i >  A/2^{2(q-1)} \ge \eta A B^2 / n^2 = \a A^{1/2} B^{3/2} /n \ge \a (2/3)c^{3/2} n.
\]
Then applying~\eqref{mixing-prop} to sets $C_i$ and $B$ (as $X$ and $Y$) would give
\begin{eqnarray*}
E(C_i,B) & \le & \frac{dC_i B}{n} + \a d \sqrt{C_i B}
   \; = \; \brac{1 + \a \frac{n}{\sqrt{C_i B}}} \frac{dC_i B}{n} \\
   & < & \brac{1 + \a 2^{i-1} \frac{n}{\sqrt{A B}}} \frac{dC_i B}{n}
   \; = \; \brac{1 + 2^{i-1} \eta} \frac{dC_i B}{n}.
\end{eqnarray*}
This would contradict~\eqref{bchwbc3}.

The expectation $\E \D_{AB}$, taken over the random selection of neighbours in
the current voting step, is equal to
\begin{eqnarray}
\E \D_{AB} & = & \sum_{v \in A} \bfrac{d_v^B}{d}^2
   \; = \; \sum_{i=0}^q \sum_{v \in C_i} \bfrac{d_v^B}{d}^2
   \; \le \; C_q + \sum_{i=0}^{q-1} \sum_{v \in C_i} \bfrac{d_v^B}{d}^2  \nonumber \\
  & \le & C_q + \sum_{i=0}^{q-1} C_i \bfrac{(dB/n)(1+2^i \eta)}{d}^2
   \; = \; C_q + \frac{B^2}{n^2} \sum_{i=0}^{q-1} C_i (1+2^{i+1} \eta + 2^{2i} \eta^2)  \label{xse319} \\
   & = & C_q + \frac{B^2}{n^2}
                     \brac{ \sum_{i=0}^{q-1} C_i + C_0 (2\eta + \eta^2)
                              + \eta\sum_{i=1}^{q-1} 2^{i+1}C_i \brac{1 + 2^{i-1} \eta}} \nonumber \\
    & \le & C_q + \frac{B^2}{n^2}
                     \brac{ A + A(2\eta + \eta^2) +
                              \eta\sum_{i=1}^{q-1} \frac{A}{2^{i-3}}\brac{1 + 2^{i-1} \eta}} \label{htd556} \\
  & \le & C_q + \frac{AB^2}{n^2}
                     \brac{ 1 + 2\eta + \eta^2 +
                              8\eta + 4q\eta^2} \nonumber \\
   & \le & C_q + \frac{AB^2}{n^2}
                     \brac{ 1 + 11\eta} \label{iow221} \\
   & \le & \frac{AB^2}{n^2} \brac{ 1 + 15\eta}.
\label{eq-upperOnDAB}
\end{eqnarray}
Inequality~\eqref{xse319} follows from the definition of sets $C_i$.
Inequality~\eqref{htd556} holds since $\sum_{i=0}^q C_i = A$,
$C_0 \le A$, and $C_i \le A/2^{2(i-1)}$.
To see that Inequality~\eqref{iow221} holds, check that 
$\eta(1+4q) \le 1$, 
using the definitions of $\eta$ and $q$, $B \ge cn$ and the bound on $\alpha$. 
Finally, Inequality~\eqref{eq-upperOnDAB} holds since
$C_q \le A/2^{2(q-1)} \le 4\eta AB^2/n^2$.

The above bound on $\E\D_{AB}$ implies
\begin{equation}
\Pr\brac{ \D_{AB} \ge  \frac{AB^2}{n^2}(1+17\eta)}
  \; \le \; \Pr\brac{ \D_{AB} \ge \brac{1+\eta} \E\D_{AB} }
  \; \le \;  e^{-\eta^2 c^2 n/8}
  \; \le \;  e^{-\a^2 c^2 n/2}. \label{eq-fger2}
\end{equation}
The second inequality above follows from~\eqref{Cher1}
applied with $\ep = \eta$, and the fact that
$\E\D_{AB} \ge (1/4)c^2n$ (from~\eqref{be56sw-Appx} with $A$ and $B$ swapped).
The last inequality in~\eqref{eq-fger2} follows from~\eqref{hjcbw671}.

{\bf Lower bound on  $\D$.}
Using~\eqref{eq-fger} and~\eqref{eq-fger2}, we get the following lower bound on the increase of the
$A$-vote:
\begin{eqnarray*}
 \lefteqn{\Pr\brac{ \D \: = \D_{BA} - \D_{AB} \; \le \;  \frac{A^2B}{n^2}(1-3\eta)
       - \frac{AB^2}{n^2}(1+17\eta)}} \\
  & \le &
\Pr\brac{ \brac{\D_{BA} \le  \frac{A^2B}{n^2}(1-3\eta)} \; \mbox{or} \;
     \brac{\D_{AB} \ge  \frac{AB^2}{n^2}(1+17\eta)}}
%\\
\; \le \; e^{-\Theta(\a^2 c^2 n)}.
\end{eqnarray*}

Hence with probability at least $1 - e^{-\Theta(\a^2 c^2 n)}$,
\begin{align*}
    \D \quad & \ge \quad \frac{AB}{n} \brac{\frac{A}{n} (1-3\eta)- \frac{B}{n}(1+17\eta)} \\
                 & \ge \quad
   \frac{AB}{n} \brac{\frac{A-B}{n}  - 10\eta}
%\\
  \;  \ge \; \frac{AB(A-B)}{n^2}  - 3 \eta n.
\end{align*}
Therefore, with probability at least $1 - Te^{-\Theta(\a^2 c^2 n)} = 1 - e^{-\Theta(\a^2 c^2 n)}$,
for all steps $t= 1,2, \ldots, T = O(\log n)$
when the size $B_t$ of the $B$-vote is at least $cn$, we have
\begin{equation}\label{bchwc}
 A_{t+1} \; = \: A_t + \D_t
\; \ge \;  A_t + \frac{A_t B_t (A_t - B_t)}{n^2}  - 3 \eta_t n.
\end{equation}
In~(\ref{bchwc}), substitute $\eta_t = \a n /\sqrt{A_t B_t}$, $A_t = n(1+\nu_t)/2$ and
$B_t = n(1-\nu_t)/2$ to get
\begin{align}
  \nu_{t+1} \;\; & \ge \; \nu_t + \frac{1}{2}\nu_t\brac{1-\nu_t^2} - 12\frac{\a}{\sqrt{1-\nu_t^2}}
       \label{ncnwd-Appx}.
\end{align}
Thus, while $K \a \le \nu_t \le 1/2$, we have
\begin{align}
  \nu_{t+1} \;\; & \ge \; \nu_t + \frac{3}{8}\nu_t - 15\a
   \; \ge \; \frac{5}{4}\nu_t. \nonumber
\end{align}
This means that the number of steps required to
increase the vote imbalance from $\nu_0$ to $1/2$ is at most
$\left\lceil \log_{5/4} (1/(2{\nu_0})) \right\rceil$.

For $1/2 \le \nu_t \le 1-2c$, we set $\d_t = 1-\nu_t$
 and~\eqref{ncnwd-Appx} becomes
\begin{align}
  \d_{t+1} \;\; & \le \; \d_t - \frac{1}{2}(1-\d_t)\d_t \brac{2-\d_t} + 12\frac{\a}{\sqrt{\d_t(2-\d_t)}}
       \label{ncnwd-Appx334}.
\end{align}
Since $2c \le \d_t \le 1/2$ and $\a \le c^{3/2}/36$,
\begin{align}
  \d_{t+1} \;\; & \le \; \d_t - \frac{3}{8}\d_t + 9\frac{\a}{\sqrt{c}}
    \; \le \; \frac{3}{4} \d_t
       \label{ncnwd-Appx334x}.
\end{align}
Thus the number of steps required to
increase the vote imbalance from $1/2$ to $1-2c$
(that is, decrease $\d_t$ from $1/2$ to $2c$) is at most
$\left\lceil \log_{4/3} (1/(4{c})) \;  \right\rceil$.
%\proofend

\section{Phase II of analysis: $\om \le B \le cn$}

\subsection{The main lemma and its application to expanders and random graphs}

The analysis of this middle phase needs 
the property that small sets of vertices do not induce many edges,
which holds for 
expanders and random regular graphs.
The main Lemma~\ref{lemPhase2} shows that for a graph with such a property,
if the minority vote is still substantial, then one voting step reduces this minority by a constant factor 
with high probability. This implies that with high probability the minority vote reduces
from $cn$ (where $c$ is a small constant) to $\omega$ within $O(\log n)$
steps (Corollary~\ref{bcwc891-cor}).

\begin{lemma}\label{lemPhase2}
Let $G$ be a $d$-regular $n$-vertex graph with $A$ and $B$ votes, 
where $A > B$.
%
%and let 
%
Let $0  < \a \le 3/10$ and
$\g = \g(\a) = (1/2)(1 - 2\a)(1 - 3\a) > 0$.
If the set $B$  is such that every superset $S \supseteq B$ of size at most $(1+1/\a)B$
spans at most $\a d S$ edges
(that is, $|E(S)| \le \a d S$), then
one voting step reduces $B$ 
by at least
a factor $1 - \g$,
with probability at least $1 - e^{-\tilde{\gamma} B}$, where 
$\tilde{\gamma}$ is some constant bounded away from $0$.
\end{lemma}

\begin{corollary}\label{bcwc891-cor}
Let $G$ be a $d$-regular $n$-vertex graph, 
and let $0 < g < 1$ be
such that for each subset of vertices $S$ of size at most $gn$,
$|E(S)| \le (3/10) d S$.
Then the minority vote $B$ is reduced from $(3/13) g n$ to at most $\omega$
within $O(\log n)$ steps with
probability at least $1-e^{-\Theta(\omega)}$.
\end{corollary}

\begin{proof}
%{\bf Proof of Corollary~\ref{bcwc891-cor}}
For each step, apply Lemma~\ref{lemPhase2} with $\a = 3/10$ and $\g = \g(3/10) > 0$.
In each step (by induction) $B \le (3/13)gn$,
so each superset $S$ of $B$ of size at most $(1 + 1/\a)B$
has size at most $(13/3)B \le gn$, implying $|E(S)| \le (3/10) dS = \a d S$.
Thus Lemma~\ref{lemPhase2} implies that
each step reduces $B$ by a factor $1-\g < 1$, with probability at least $1 - e^{-\tilde{\g} B}$.

Let $r = \lceil \log (cn/\om) / \log (1/(1-\g)) \rceil = O(\log n)$, so that 
$\om < (1-\g)^{r-1} cn$ but $\om \ge (1-\g)^{r} cn$.
The initial size of the minority vote $B$ is $B_0 \le cn$.
We say that step $i$ is successful, if the size of the $B$ vote at the 
end of this step is $B_i \le (1-\g)^i cn$.
If the steps $1,2,\ldots, i-1$ are succesful, then $B_{i-1} \le (1-\g)^{i-1} cn$
and the probability that $B_i \le  (1-\g)^{i} cn$ (that is, the probability that step $i$ is 
successful) is at least the probability that 
a $B$ vote of size $(1-\g)^{i-1} cn$ reduces in one step to $(1-\g)^{i}cn$.
This (conditional) probability is at least $1- \exp\{-\tilde{\g} (1-\g)^{i-1} cn \}$ (Lemma~\ref{lemPhase2}).
Therefore
\begin{eqnarray}
 \Pr\brac{B_r \le \om} 
    & \ge & \Pr\brac{\mbox{all steps $1,2, \ldots, r$ are successful}} \nonumber \\
    & = & \Pi_{i=1}^r \Pr\brac{\mbox{step $i$ is successful}\: |\: \mbox{steps $1,2,\ldots,i-1$ are successful} }
                 \nonumber \\
     & \ge & \Pi_{i=1}^r \brac{1 - \exp\{-\tilde{\g} (1-\g)^{i-1} cn\}} \\
     & = & 1 - e^{- \Theta\brac{\om}}.
\end{eqnarray}
Thus with probability at least
$1 - e^{-\Theta(\omega)}$,
$B$ is reduced from $cn$ to $\omega$ in $O(\log n)$ steps.
\end{proof}

\begin{lemma}\label{phaseII-expander}
Let $G$ be a $d$-regular $n$-vertex graph with
$\l = \l_G < 3/5$.
Then the minority vote $B$ reduces from $(3/13)(3/5 - \l)n$
to $\omega$ within $O(\log n)$ steps with
probability at least $1-e^{-\Theta(\omega)}$.
\end{lemma}

\begin{proof}
This lemma follows from Corollary~\ref{bcwc891-cor}
applied with $c = 1 - (2/5)(1-\l)^{-1} > 0$, after checking that $|E(S)| \le (3/10)dS$ whenever $S \le cn$.
It is shown in~\cite{Jerrum-Sinclair} that the conductance of graph $G = (V,E)$ defined as
\[ \Phi_G = \min_{\emptyset \neq S\subset V} \frac{n E(S,\bar{S})}{d S \bar{S}},
\]
is at least $1 - \l$.
This implies that 
$E(S,\bar{S}) \ge (1 - \l) dS\bar{S}/n$, 
so for any $S\subseteq V$ of size at most~$cn$,
\begin{eqnarray*}
    |E(S)| & = & \frac{1}{2} \brac{dS - E(S,\bar{S})} 
    \; \le \; \frac{1}{2} d S \brac{1 - (1-\l) \bar{S}/n} \\
    & \le & \frac{1}{2} d S \brac{1 - (1-\l)(1-c)}
    \; = \; \frac{3}{10} dS.
\end{eqnarray*}
\end{proof}

As mentioned before, for constant $d$ 
a random $d$-regular graph has eigenvalue $\lambda_G \approx 2/\sqrt{d}$, \whp\
Thus,
for constant $d$, the result which we have obtained for expanders (Lemma~\ref{phaseII-expander}) applies 
to random regular graphs as well,
provided that $d$ is sufficiently large to guarantee $\lambda_G < 3/5$. 
To consider 
random regular graphs with degree 
which may grow with the number of vertices, 
%larger than $600$.
we show the following lemma.
This is a stronger version of a result from \cite{CF} that \whp\ for
$3 \le d \le cn$ no set of vertices of size  $|S| \le n/70$ induces more than $d|S|/12$ edges.

\begin{lemma}\label{span}
Let $600 \le d \le n/K$ for some large constant $K$,
and let $G = (V,E)$ be a random $d$-regular $n$ vertex graph.
Let $\a=1/12$ and consider the event
\[
\ul Q = \set{\exists S \seq V: |S| \le n/15 \text{ and $S$ spans at least } \a d|S| \text{ edges }}.
\]
Then $\Pr(\ul Q) \le  n^{-\d}$, for some constant $\d > 0$.
\end{lemma}

\begin{lemma}\label{phaseII-randomgraph}
Let $ d > 10$ and $G$ be a random $d$-regular $n$-vertex graph
with votes $A$ and $B$.  
There is a constant $c >0$ (independent of $d$) such that 
the minority vote $B$ reduces from $cn$
to $\omega$ within $O(\log n)$ steps with
probability at least $1-e^{-\Theta(\omega)} - o(1/n)$.
\end{lemma}

\begin{proof}
For $11 \le d < 600$ use Lemma~\ref{phaseII-expander}:
in this case, a random $d$-regular graph has $\l_G \le (2\sqrt{d-1} + \epsilon)/d < 3/5$.

For $600 \le d \le n/K$, Lemma~\ref{span} implies that
with probability at least $1 - o(1/n)$, $E(S) \le (1/12) dS < (3/10)dS$, for each subset of vertices
$S$ of size at most $n/15$.
Thus, applying Corollary~\ref{bcwc891-cor} with $g = 1/15$,
we conclude that the minority vote $B$ is reduced from $(3/13) g n = (1/65)n $ to at most $\omega$
within $O(\log n)$ steps with
probability at least $(1-o(1/n))(1-e^{-\Theta(\omega)})$.
\end{proof}

\subsection{Proof of Lemma~\ref{lemPhase2}}

Consider first the following special case.
For each vertex $v\in B$, $d_v^A = (1-2\a)d$ (so $|E(B)| = \b d B$),
and for each $v\in A$, $d_v^B$ is either $0$ or $\a d$.
Since $\sum_{v\in A} d_v^B = \sum_{v\in B} d_v^A$, then
the number of vertices $v$ in $A$ with $d_v^B = \a d$ is
$B(1-2\a)/\a$, so in this case
the expected increase of the $A$ vote is equal to
\[ \E\D \; = \; (1-2\a)^2 B - \a^2 B(1-2\a)/\a
      \; = \; (1-2\a)(1-3\a)B.
\]
The proof of Lemma~\ref{lemPhase2} is based on confirming that this is the worst case,
that is, we always have
\begin{equation}\label{ngsdjkfbi819}
 \E\D \; \ge \; (1-2\a)(1-3\a)B.
\end{equation}

Let $|E(B)| = \b d B$. Since we must have $|E(B)| \le \a dB$, then
$0 \le \b \le \a$.
We have
\[ \sum_{v\in A} d_v^B \; = \; \sum_{v\in B} d_v^A \; = \; dB - 2|E(B)| \; =\; dB(1-2\b),
\]
so
\begin{equation}\label{bwhecwvc1}
\E \D_{BA} = \sum_{v \in B} \bfrac{d_v^A}{d}^2
    \; \ge \; \frac{1}{Bd^2} \brac{\sum_{v \in B} d_v^A}^2
    \; = \; \frac{1}{Bd^2} \brac{dB - 2|E(B)|}^2
    \; = \; B(1-2\b)^2.
\end{equation}

To bound $\E\D_{AB}$, we define $C = \{ v \in A: d_v^B > \a d \}$.
We have
\[ \a d C \; \le \; |E(C,B)| \; \le \; dB,
\]
so $C+B \le (1 + 1/\a)B$, and the assumptions of the lemma imply that
\begin{equation}\label{ncwehc901-asa}
 |E(C\cup B)| \le \a d(C+B).
\end{equation}

For $v \in C$, we write $d_v^B$ as a linear combination of $d$ and $\a d$:
\[ d_v^B = x_v d + y_v (\a d); \;\;\; x_v + y_v = 1; \;\;\; x_v, y_v \ge 0.
\]
For $v \in A\setminus C$, we define
$y_v = d_v^B/(\a d)$, so $0 \le y_v \le 1$, and set $x_v = 0$.
We also define $X = \sum_{v\in A} x_v$,
$Y = \sum_{v\in A} y_v$, and
$Y_C = \sum_{v\in C} y_v$.
Using this definitions,
\[ dB(1-2\b) = \sum_{v\in A} {d_v^B} = dX + \a d Y,
\]
so
\begin{equation}\label{bbjqs781}
 \a Y = B(1-2\b) - X.
\end{equation}
Furthermore,
\begin{equation}\label{ncwehc901}
 |E(C\cup B)| = \sum_{v\in C} d_v^B + \b dB = dX + \a d Y_C + \b dB,
\end{equation}
and
\begin{equation}\label{ncwehc901-xx}
 C = X + Y_C.
\end{equation}
Using~\eqref{ncwehc901} and~\eqref{ncwehc901-xx} in~\eqref{ncwehc901-asa}, we get
\[ X + \a  Y_C + \b B \le \a(X + Y_C +B),
\]
so
\begin{equation}\label{bbjqs781-a}
 X \le \frac{\a - \b}{1-\a} B.
\end{equation}
We use~\eqref{bbjqs781} and~\eqref{bbjqs781-a} in bounding $\E\D_{AB}$:
\begin{eqnarray}
 \E \D_{AB} &  = & \sum_{v \in A} \bfrac{d_v^B}{d}^2
      \; = \; \sum_{v \in C} (x_v + \a y_v)^2 + \sum_{v \in A\setminus C} (\a y_v)^2 \nonumber \\
      & \le & \sum_{v \in C} (x_v + \a^2 y_v)  + \sum_{v \in A\setminus C} \a^2 y_v \nonumber \\
      & = & X + \a^2 Y
      \; = \; X + \a (B(1-2\b) - X)
      \; = \; X(1-\a) + B\a (1-2\b) \nonumber \\
      & \le & B(\a - \b) + B\a (1-2\b)
      \; = \; B(2\a - \b - 2\a\b). \label{bwhecwvc2}
\end{eqnarray}
The bounds~\eqref{bwhecwvc1} and~\eqref{bwhecwvc2} give
\[ \E\D \; = \; \E \D_{BA} - \E \D_{AB}
     \; \ge \; B((1-2\b)^2 - 2\a + \b + 2\a\b)
     \; = \; B \cdot f_\a(\b),
\]
where
\[ f_\a(\b) = 4\b^2 - (3-2\a)\b + (1-2\a).\]
We check that $f'_\a(\b) = 8\b - 3 + 2\a \le 10\a -3 \le 0$, for $0 \le \b \le \a \le 3/10$, so
the minimum value of $f_\a(\b)$ for $0 \le \b \le \a$ is
$f_\a(\a) = (1 - 2\a)(1 - 3\a) = 2\g$, and the bound~\eqref{ngsdjkfbi819} holds.

\ignore{
As long as $B = \omega(\sqrt{n \log n})$, we can apply the Hoeffding inequality to 
show the statement of the lemma. However, for the case $B = O(\sqrt{n \log n})$ this inequality does not imply 
our result. Therefore, we bound $\D_{AB}$ and $\D_{BA}$ separately.
}

Thus $\E\D \ge 2\g B$ and we show now that $\D$ is at least $\g B$ with high probability,
by showing that \whp\ $\D_{AB}$ and $\D_{BA}$ do not deviate from their expectations by 
more than $\g B/2$.
For $\D_{BA}$, using~(\ref{Cher1_new}) with $\e = \e_1 = \g/(2(1-2\b)^2) < 1$, we obtain
\begin{eqnarray} \label{eq_Chernoff_DBA}
\Pr\brac{\D_{BA} \le (1-2\b)^2 B - \frac{\g}{2} B} 
 & = & \Pr\brac{\D_{BA} \le \brac{1 - \frac{\g}{2 (1-2\b)^2}} (1-2\b)^2 B} \nonumber \\
 & = & \Pr(\D_{BA} \le \brac{1 - \e_1} (1-2\b)^2 B) %\nonumber \\
 \;\; \le \;\;   e^{-\gamma' B},  \label{xmsciqps}
\end{eqnarray}
where $\gamma' = (\g/(2(1-2\beta)))^2 / 3 >0$ is a constant.

For $\D_{AB}$, 
$\e = \e_2 = \g/(2(2\alpha-\beta-2\alpha \beta))$.
%
%we use~(\ref{Cher1}) with 
%$\e = \e_2 = \min\{ \g/(2(2\alpha-\beta-2\alpha \beta)),\: 1\}$.
%Here we take the minimum with $1$ to ensure that $\e \le 1$.
%We obtain
%\begin{eqnarray} \label{eq_Chernoff_DAB}
%\lefteqn{\Pr\brac{\D_{AB} \geq (2\alpha-\beta-2\alpha \beta) B + \frac{\g}{2}B}} \nonumber \\
% & = & \Pr\brac{\D_{AB} \ge \brac{1 + \frac{\g}{2(2\alpha-\beta-2\alpha \beta) }} 
%                 (2\alpha-\beta-2\alpha \beta) B} \nonumber \\
% & \le & \Pr(\D_{AB} \ge \brac{1 + \e_2}(2\alpha-\beta-2\alpha \beta) B) %\nonumber \\
% \;\; \le \;\;   e^{-\gamma'' B},  \label{njsnq1}
%\end{eqnarray}
%where $\gamma'' = \e_2^2 (2\alpha-\beta-2\alpha \beta)/2 >0$ is a constant.
%The bounds~\eqref{xmsciqps} and~\eqref{njsnq1} imply
%\begin{eqnarray*}
%  \lefteqn{\Pr(\D \le \g B)} \\ 
%    & \le & \Pr\brac{\D_{BA} - \D_{AB} \: \le \: ((1-2\b)^2 B - \frac{\g}{2}B) 
%          - ((2\alpha-\beta-2\alpha \beta) B - \frac{\g}{2}B)} \\
%   & \le & \Pr\brac{\D_{BA} \le (1-2\b)^2 B- \frac{\g}{2} B }
% + \Pr\brac{\D_{AB} \geq (2\alpha-\beta-2\alpha \beta) B + \frac{\g}{2}B}
% \le e^{-\tilde{\gamma}B}.
%\end{eqnarray*}
%
%\ignore{
Now we cannot guarantee that 
$\epsilon <1$. As long as $(1-2\alpha)(1-3\alpha) < 4(2\alpha-\beta-2\alpha \beta)$, we apply~(\ref{Cher1}) and obtain
\begin{eqnarray*} \
\Pr\brac{\D_{AB} \ge (2\alpha - \beta-2 \alpha \beta) B + \frac{\g}{2} B} 
 %& = & \Pr\brac{\D_{BA} \le \brac{1 - \frac{\g}{2 (1-2\b)^2}} (1-2\b)^2 B} \nonumber \\
 %& = & \Pr(\D_{BA} \ge \brac{1 - \e_2} ( B) %\nonumber \\
&=& \Pr(\D_{AB} \geq (1+\epsilon_2) (2\alpha-\beta-2\alpha \beta) B) \\
&\leq& \exp\left\{ -\frac{(1-2\alpha)^2(1-3\alpha)^2}{16(2\alpha-\beta-2\alpha \beta)} \cdot 
\frac{B}{2}\right\} =
e^{-\gamma'' B} ,
\end{eqnarray*}
where $\gamma'' = \frac{(1-2\alpha)^2(1-3\alpha)^2}{32(2\alpha-\beta-2\alpha \beta)} >0$ is a constant. If now $\epsilon \geq 1$, we 
apply (\ref{Cher2}), and obtain 
\begin{eqnarray*} 
\Pr\brac{\D_{AB} \ge (2\alpha - \beta-2 \alpha \beta) B + \frac{\g}{2} B} 
&=&\Pr(\D_{AB} \geq (1+\epsilon_2) (2\alpha-\beta-2\alpha \beta) B) \\
&\leq& \left(\frac{e^\epsilon}{(1+\epsilon)^{1+\epsilon}}  \right)^{(2\alpha-\beta-2\alpha \beta)B}
%\exp\left\{ -\frac{(1-2\alpha)^2(1-3\alpha)^2}{16(2\alpha-\beta-2\alpha \beta)} \frac{B}{2}\right\} =
= e^{-\gamma''' B} ,
\end{eqnarray*}
where $\gamma''' = \ln((1+\epsilon)^{(1+\epsilon)/\epsilon}/e) (1-2\alpha)(1-3\alpha)/4$. 
%In all cases, 
%$\D \ge (1/2)(1 - 2\a)(1 - 3\a) B$ with probability at least $1-e^{-\min\{\gamma', \gamma'', \gamma'''\} B}$.
%%}

The bounds above imply that 
\begin{eqnarray*}
 \lefteqn{\Pr(\D \le \g B)} \\ 
   & \le & \Pr\brac{\D_{BA} - \D_{AB} \: \le \: ((1-2\b)^2 B - \frac{\g}{2}B) 
         - ((2\alpha-\beta-2\alpha \beta) B + \frac{\g}{2}B)} \\
  & \le & \Pr\brac{\D_{BA} \le (1-2\b)^2 B- \frac{\g}{2} B }
 + \Pr\brac{\D_{AB} \geq (2\alpha-\beta-2\alpha \beta) B + \frac{\g}{2}B}
 \le e^{-\tilde{\gamma}B}.
\end{eqnarray*}

%\proofend

\subsection{Proof of Lemma~\ref{span}}
Let $S$ also denote the size of set $S$.
If $S \le 2\a d$, then
the set $S$ has ${s \choose 2} < \a d S$ slots for edges, so it is not possible to insert
$\a d S$ edges into $S$. We therefore assume that $2\a d +1 \le S \le n/15$.

The proof for this case is in the configuration model.  Let $\r(S,k)$ be the probability
a set size $S$ spans at least $k$ edges. Then, using the notation $m!!=m!/(2^{m/2}(m/2)!)$,
\[
\r(S,k) = {dS \choose 2k} 2k!! \frac{(dn-2k)!!}{dn!!}.
\]
After some simplification, we have that
\[
\r(S,k) = O(1) \frac{(dS)^{dS}}{k^k 2^k (dS-2k)^{dS-2k}} \frac{(dn-2k)^{(dn-2k)/2}}{dn^{dn/2}}.
\]
To simplify further, let $k=\a d S$, and $S=sn$. Then
\begin{align*}
\r(S,k)=& O(1) \brac{ \bfrac{s}{2\a}^{\a s}\frac{(1-2\a s)^{1/2-\a s}}{(1-2\a)^{s(1-2\a)}}}^{dn}\\
&= O(1) [f(s,\a)]^{dn}.
\end{align*}
Thus, the probability that a random $d$ regular $n$-vertex graph has the
property $\ul Q$ is 
\begin{align*}
\Pr(\ul Q)  \le& \; O(1) \cdot e^{20d^2} \cdot 
       \sum_{S=d/6}^{n/15} {n \choose S} [f(S/n, 1/12)]^{dn} \\
=& \; O(1) \cdot e^{20d^2} \cdot \sum \brac{ \frac{1}{s^s(1-s)^{1-s}}
\left[ (1-s/6)^{1/2} \brac{ \frac{36 s}{6-s}  \bfrac{6}{5}^{10}}^{s/12}\right]^d}^n \\
=& \; O(1) \cdot e^{20d^2} \cdot \sum [F(s)]^n.
\end{align*}
The term $e^{20 d^2}$ takes into account that we only consider simple graphs.
%
%\begin{align*}
%\Pr_{\cC}(\ul Q)  \le& \; O(1) \sum_{S=d/6}^{n/15} {n \choose S} [f(S/n, 1/12)]^{dn}\\
%=& \; O(1) \sum \brac{ \frac{1}{s^s(1-s)^{1-s}}
%\left[ (1-s/6)^{1/2} \brac{ \frac{36 s}{6-s}  \bfrac{6}{5}^{10}}^{s/12}\right]^d
%}^n\\
%=& \; O(1) \sum [F(s)]^n.
%\end{align*}
%
$F(x)$ can be written as
\begin{align*}
F(x) = & \frac{1}{x^x(1-x)^{1-x}} \bfrac{6}{5}^{5dx/6}\frac{ 6^{dx/6}}{6^{d/2}} \left[(6-x)^{6-x} x^x\right]^{d/12}.
\end{align*}
The second derivative of $\log F(x)$ is
\[
\frac{\partial^2 }{\partial x^2}\log F(x) = \frac{1}{x}\brac{\frac{d}{12}-1} +\frac{d}{12} \frac{1}{6-x} -\frac{1}{1-x}.
\]
Provided $d \ge 132$, this is strictly greater than zero for all $x\in(0,1/2)$.
Thus $\log F(x)$ is convex, and is either
monotone increasing, monotone decreasing, or has a unique minimum in the range $x \in [d/6n, 1/15]$.
To find the maximum of $F(x)$ it suffices to evaluate the function at
$x=d/6n$ and $x= 1/15$. Thus, assuming $d \le n/K$ for a large constant $K$,
\begin{align*}
\Pr(\ul Q)  \le& \; O(n) \max \brac{[F(d/6n)]^n, [F(1/15)]^n} \cdot e^{20d^2}.
%=&O(n) [F(1/15)]^n.
\end{align*}
However $F(1/15)=((0.999514)^d/0.782759)$, which gives $F(1/15)=0.954335$ when $d=600$. Thus for
$d \ge 600$ we have $[F(1/15)]^{n}  \le e^{-\Theta(dn)}$.
To bound  $[F(d/6n)]^n$, observe that $F(x) \le (c x)^{x(d/12 - 1)}$, for some constant $c > 0$ and
for all $x\in (0,1/2)$.
Hence, for a positive $d \le n/K$, we have $[F(d/6n)]^n \le e^{-\Theta(d^2 \log(n/d))}$, so
for $600 \le d \le n/K$,
\begin{align*}
\Pr_{}(\ul Q)  \le& \; \; e^{-\Theta(d^2 \log(n/d))}
 \; \le \; n^{-\d},
\end{align*}
where $\d >0$ is a constant.
%
%\begin{align*}
%\Pr_{\cG}(\ul Q)  \le& \; \Pr_{\cC}(\ul Q) \cdot e^{20d^2} \; \le \; e^{-\Theta(d^2 \log(n/d))}
% \; = o\bfrac{1}{n}.
%\end{align*}
%

%\end{proof}
%\proofend

\section{Phase III of analysis: $1 \le B \le \om$}

\ignore{
\begin{lemma}
[\bf Case $d \rai$ with $n$]
If $d \rai$ with $n$, $G$ is any $n$-vertex $d$-regular graph with the initial size of the $B$-vote
at most $d^{1/4}$, then \whp\ in one step $B$ is empty.
\end{lemma}

\begin{proof}
With the conditions of the lemma, for each vertex $v$ of $A$, $d^B_v \le B \le d^{1/4}$,
and there are at most $Bd$ edges between sets $A$ and $B$,
so the expectation of the number of vertices which change in one round from $A$ to $B$ is
\[
\E \D_{AB} \; = \;  \sum_{v \in A} \bfrac{d_v^B}{d}^2
  \; \le \;  \sum_{v \in A} \frac{B d_v^B}{d^2}
  \; = \; \bfrac{B^2}{d} \le \frac{1}{d^{1/2}}.
\]
On the other hand, any vertex $v \in B$ has $d^A_v \ge d-B$, so $\E \D_{BA} \ge B(1-2B/d)$.
Let
$B'= B-\D_{BA}+\D_{AB}$
be the size of $B$ at the end of the round.
\[
\E B' \; \le \; \frac{2B^2}{d}  + \frac{1}{d^{1/2}}
\; \le \; \frac{2}{\sqrt{d}}+\frac{1}{d^{1/4} }
\; \le \; \frac{3}{d^{1/2}},
\]
so that
\[
\Pr( B' \ge 1) \; \le \; \frac{3}{d^{1/2}}.
\]
\end{proof}
}

\begin{lemma}
\label{d-constPhaseIII}
Let 
$\om = \om(1)$
%
%$\om = \om(n)$ 
%
grow with $n$ and 
$\om = o(n)$.
%
%$\om(n) = o(n)$.
%
Let $G$ be a $d$-regular $n$-vertex graph
such that for each subset of vertices $S$ of size at most $(13/3)\om$,
$|E(S)| \le (3/10) d S$.
Then the minority vote $B$ is reduced from $\om$ to $0$
within $O(\om\log \om)$ steps with
probability at least $1-e^{-\Theta(\om)}$.
\end{lemma}

\begin{proof}
Lemma \ref{lemPhase2}
implies that if $B = \omega$, then 
the probability that the next step increases $B$ above $\om$ is less than 
$e^{-\tilde{\g}\om}$, where $\tilde{\g} >0$.
%
%$\g = \g(3/10) > 0$.
%
This also implies that if $B < \om$, then 
the probability that the next step increases $B$ above $\om$ is also less than 
$e^{-\tilde{\g}\om}$ (smaller $B$ means smaller probability that one step will
increase this $B$ above $\om$).
Therefore, the probability that $B$ increases above $\om$ in any of the $T= O(\om \log \om )$
steps is at most $Te^{-\tilde{\g}\om} = e^{-\Theta(\om)}$.

Let $r = \lfloor \log \om / \log (1/(1-\g)) \rfloor +1 = O(\log \om)$, so that 
$(1-\g)^{r-1} \om \ge 1$ but $(1-\g)^{r} \om < 1$.
The initial size of the minority vote $B$ is $B_0 \le \om$.
We say that step $i$ is successful, if the size of the $B$ vote at the 
end of this step is $B_i \le (1-\g)^i \om$.
If the steps $1,2,\ldots, i-1$ are successful, then $B_{i-1} \le (1-\g)^{i-1} \om$
and the probability that $B_i \le  (1-\g)^{i} \om$ (that is, the probability that step $i$ is 
successful) is at least the probability that 
a $B$ vote of size $(1-\g)^{i-1} \om$ reduces in one step to $(1-\g)^{i}\om$,
which is at least $1-\exp\{-\tilde{\g} (1-\g)^{i-1} \om \}$ (Lemma~\ref{lemPhase2}).
Therefore
\begin{eqnarray}
 \Pr\brac{B_r = 0} 
    & \ge & \Pr\brac{\mbox{all steps $1,2, \ldots, r$ are successful}} \nonumber \\
     & \ge & \Pi_{i=1}^r \brac{1 - \exp\{-\tilde{\g} (1-\g)^{i-1} \om\}}
     \; = \; p > 0,
\end{eqnarray}
where $p$ is a positive constant.
Thus $B$ is reduced from $\om$ to $0$ within $r = O(\log \om)$ steps with constant (positive) probability.
Consider now a sequence of $\om r = O(\om \log \om)$ steps, 
viewed as $\om$ phases, each consisting of $r$ steps.
If this sequence of steps does not reduce $B$ from $B_0 \le \om$ to $0$, then
there is a step which increases $B$ above $\om$ or
each phase starts with $B \le \om$ but fails to reduce $B$ to $0$. 
This means that the probability that $\om r$ steps do not reduce $B_0 \le \om$ to $0$ 
is at most $e^{-\Theta(\om)} + (1-p)^\om = e^{-\Theta(\om)}$.
\end{proof}

\begin{corollary}\label{nvkq72}
Let
$\om = \om(1)$
%
%$\om = \om(n)$ 
%
grow with $n$ and 
$\om = o(n)$.
%
%$\om(n) = o(n)$.
%
If $G=(V,E)$ is a $d$-regular expander with 
$\lambda_G < 3/5$
or it is a random $d$-regular graph with
$d > 10$, then voting reduces $B$ from at most $\om$ to $0$ in $O(\om \log \om)$ steps
with probability at least $1-e^{-\Theta(\om)}$.

\end{corollary}
\begin{proof}
If $G$ is a $d$-regular expander with $\lambda_G \leq 3/5$, then the assumptions of Lemma
\ref{d-constPhaseIII} are fulfilled for $G$, as shown in the proof of Lemma~\ref{phaseII-expander}. 
If $G$ is a random $d$-regular graph with $d = \omega(1)$, then
the assumptions of Lemma \ref{d-constPhaseIII} are also fulfilled for $G$ according to Lemma \ref{span}.
If $d>10$ but $d= O(1)$, then $G$ has eigenvalue $\lambda_G < 3/5$, \whp~\cite{Fri03}.
\end{proof}

\section{Putting the phases together}

To conclude the proof of our main Theorems~\ref{Th1} and~\ref{Th1-expanders},
it remains to check how the three phases fit together.
For expanders (Theorem~\ref{Th1-expanders}),
first use Corollary~\ref{lemMethod2-expanders} with $c = 1/10$
to get constant $K = K(c)$ such that
if the initial imbalance of vote is $\nu_0 \ge K\l_G$,
then the minority vote reduces to $n/10$ within $O(\log (1/\nu_0))$ steps.
Then use Lemma~\ref{phaseII-expander}
with $\om = \log n / \log\log n$ and assume that $\l_G \le 1/6$
to show that the minority vote reduces from $n/10$ to $\om$ in $O(\log n)$ steps.
Finally, apply Corollary~\ref{nvkq72} with the same $\om$ to show 
that the minority vote decreases from $\om$ to $0$ in $O(\log n)$ steps.

For the random regular graphs,
Lemma~\ref{phaseII-randomgraph} 
gives the constant $c < 1/2$ which defines
the beginning of phase II.
Then Corollary~\ref{lemMethod2-randomgraphs} can be used to 
find the constant $K$ for Theorem~\ref{Th1}.
The transition from phase II to phase III is at the same $\om = \log n / \log\log n$ 
as before.

According to our analysis, we can also derive the following corollary.
\begin{corollary}
\label{robustness}
Assume an adversary can change the opinion of at most $f= o(\nu_0 n)$ vertices
during the execution of the algorithm. Then, under the assumptions of Theorems \ref{Th1}
and \ref{Th1-expanders}
%w.r.t.~the initial sets $A$ and $B$,
all but $O(f)$ vertices will adopt
opinion $A$ within $O(\log n)$ steps, \whp
\end{corollary}

To obtain the statement, observe that the assumptions in
Phases I, II, and III are fulfilled w.r.t.~$A-f$ and $B+f$ as long as $B \geq C' \cdot f$,
where $C'$ is a suitable large constant.


\newpage

\begin{thebibliography}{99}

\bibitem{AD} M. Abdullah and M. Draief. {\em Consensus on the Initial Global Majority by Local Majority Polling
for a Class of Sparse Graphs.} (2013)  {\tt www.arXiv.org}

\bibitem{AlFi} D. Aldous and J. Fill. {\em Reversible Markov Chains and
Random Walks on Graphs}, \\
{\tt http://stat-www.berkeley.edu/pub/users/aldous/RWG/book.html.}

\bibitem{ExpanderLemma}
N. Alon and F. R. K. Chung. {\em Explicit construction of linear sized tolerant networks.}
Discrete Math., 72:15-19, (1989).

%\bibitem{AS} N. Alon and J. Spencer, {\em The Probabilistic Method}, Second
%  Edition, Wiley-Interscience, (2000).

\bibitem{BCNP+13}
L.~Becchetti, A.~Clementi, E.~Natale, F.~Pasquale, R.~Silvestri, L.~Trevisan.
{\em Simple Dynamics for Majority Consensus.} (2013) {\tt www.arXiv.org}

\bibitem{Bol88}
B. Bollob{\'a}s. {\em The isoperimetric number of random regular graphs.}
Europ.~J.~Combinatorics, 9:241-244, (1988).

\bibitem{BMPS04}
S. Brahma, S. Macharla, S. P. Pal, S. R. Singh. {Fair Leader
Election by Randomized Voting.}
In {\em ICDCIT 2004},
% Proceedings of the 1st International
%Conference on Distributed Computing and Internet Technology}, 
pages 22-31, 2004.

\bibitem{CEHR12}
C. Cooper, R. Els{\"a}sser, H. Ono, T. Radzik. {Coalescing Random Walks
and Voting on Graphs.}
In {\em PODC 2012},
%Proceedings of the ACM Symposium
%on Principles of Distributed Computing}, 
pages 47-56, 2012.
%full version available at {\tt arxiv.org/pdf/1204.4106.pdf}.

\bibitem{CEHR-SIAM2013}
C. Cooper, R. Els{\"a}sser, H. Ono, T. Radzik. 
{\em Coalescing Random Walks and Voting on Connected Graphs.}
SIAM J. Discrete Math. 27(4):1748-1758, (2013).

\bibitem{XXX}
C. Cooper, A. Frieze, B. Radzik. {\em Multiple Random Walks
in Random Regular Graphs.}
SIAM J. Discrete Math. 23(4):1738-1761, (2009).

\bibitem{CF}
C. Cooper, A. Frieze, B. Reed. {\em Random regular
graphs of non-constant degree: connectivity and Hamilton cycles.}
Combinatorics Prob. \& Comp. 11:249-262, (2002).

\bibitem{CG} J. Cruise and A. Ganesh,
{\em Probabilistic consensus via polling and majority
rules.} (2013)
{\tt www.arXiv.org}.

\bibitem{DP94}
X. Deng and C. Papadimitriou. {\em On the Complexity of Cooperative
Solution Concepts.}
Mathematics of Operations Research 19(2):257-266, (1994).


\bibitem{DGMSS} B.~Doerr, L.A.~Goldberg, L.~Minder, T.~Sauerwald,
C.~Scheideler:
{Stabilizing Consensus with the Power of Two Choices.} In
{\em SPAA 2011}
%Proceedings of the 23rd Annual ACM Symposium
%on Parallelism in Algorithms and Architectures},
pages 149-158, 2011.
%full version available at {\tt www.upb.de/cs/scheideler}.


\bibitem{DonnellyWelsh-1983} P. Donnelly and D. Welsh.
{\em Finite particle systems and infection models.}
Math. Proc. Camb. Phil. Soc.  94(1):167-182, (1983).

\bibitem{FountoulakisPanagiotou}
N. Fountoulakis  and K. Panagiotou.
{\em Rumor Spreading on Random Regular Graphs and Expanders.}
{\em APPROX and RANDOM 2010},
%Approximation, Randomization, and Combinatorial Optimization. Algorithms and Techniques},
%{LNCS v. 6302}, 
pages {560-573}, 2010.

\bibitem{FL} A. Frieze and T. {\L}uczak.
{\em On the independence and chromatik numbers of random graphs.}
J. Combinatorial Theory, Ser. B,  54:123-132, (1992).

\bibitem{Fri03} J. Friedman
{A proof of {A}lon's second eigenvalue conjecture.}
In {\em STOC 2003,}
%
%In {\em STOC 2003: Proc. 35th Annual ACM Symposium on
%Theory of Computing}, 
%
pages 720-724, 2003.

\bibitem{Gifford79} D. Gifford.
Weighted Voting for Replicated Data.
In {\em SOSP 1979,}
%
%In {\em SOSP 1979: Proceedings of the 7th ACM Symposium on Operating Systems Principles},
%
pages 150-162, 1979.

\bibitem{HassinPeleg-InfComp2001}
Y.~Hassin and D.~Peleg.
{\em Distributed probabilistic polling and applications to proportionate agreement.}
Information \& Computation, 171(2):248-268, (2001).

\bibitem{Jerrum-Sinclair}
M.~Jerrum and A.~Sinclair.
{Conductance and the rapid mixing property for Markov chains: the approximation of permanent resolved.}
In {\em STOC 1988,}
%
%In {\em STOC 1988: Proceedings of the 20th Annual ACM symposium on Theory of Computing},
%
pages 235-244, 1988.

\bibitem{Joh89} B. Johnson, {\em Design and Analysis of Fault Tolerant Digital Systems},
Addison-Wesley, (1989).


\bibitem{MNT}
E. Mossel , J. Neeman, O. Tamuz, {\em Majority Dynamics and Aggregation of Information in
Social Networks.} (2012) {\tt www.arXiv.org}.


\bibitem{Nakata_etal_1999}
T. Nakata, H. Imahayashi, M. Yamashita.
Probabilistic local majority voting for the agreement problem on
finite graphs.
In {\em COCOON 1999},
% Proc. 5th Annual International
%Conference on Computing and Combinatorics}, 
pages
330-338, 1999.

\bibitem{Oli12}
R.I. Oliviera. {\em On the Coalescence Time of Reversible Random Walks.}
Trans. Amer. Math. Soc., 364:2109--2128, (2012).

\bibitem{Wormald-1999}
N.C.~Wormald. Models of random regular graphs.
In {\em Surveys in Combinatorics} (J. D. Lamb and D. A. Preece, eds), pp. 239--298.

\end{thebibliography}
\end{document}